\documentclass[12pt]{amsart}

\usepackage{amsmath}
\usepackage{amssymb}
\usepackage{graphicx}
\usepackage{pifont}
\usepackage{epsfig}
\usepackage{verbatim}
\usepackage[english]{babel}
\usepackage{fancyhdr, blindtext}
\usepackage{float}
\usepackage{subfig}

\usepackage{graphics}
\usepackage{graphicx,epsfig}
\usepackage{epsfig}
\usepackage{graphicx}

\usepackage{a4wide}
\usepackage{amsmath}
\usepackage{amsfonts}

\usepackage{enumerate}
\usepackage{authblk}
\usepackage{float}
\usepackage{subfig}
\usepackage{amsaddr}

\begin{document}

\newcommand\bes{\begin{eqnarray}}
\newcommand\ees{\end{eqnarray}}
\newcommand\bess{\begin{eqnarray*}}
\newcommand\eess{\end{eqnarray*}}
\newcommand{\ve}{\varepsilon}
\newtheorem{definition}{Definition}
\newtheorem{theorem}{Theorem}[section]
\newtheorem{lemma}{Lemma}[section]
\newtheorem{proposition}{Proposition}[section]
\newtheorem{remark}{Remark}[section]
\newtheorem{corollary}{Corollary}[section]
\newtheorem{example}{Example}[section]
\title[Analysis of transient dynamics ]
{\bf Analysis of long-term transients and detection of early warning signs of major population changes in a two-timescale ecosystem.}

\author{Susmita Sadhu}
\address{Department of Mathematics,
Georgia College \& State University, Milledgeville, GA 31061, USA}
\email{susmita.sadhu@gcsu.edu}


\thispagestyle{empty}


\begin{abstract}
\noindent    Identifying early warning signs of sudden population changes and mechanisms leading to regime shifts are highly desirable in population biology. In this paper, a two-trophic ecosystem comprising of two species of predators, competing for their common prey, with explicit interference competition is considered. With proper rescaling, the model is portrayed as a singularly perturbed system with fast prey dynamics and slow dynamics of the predators. In a parameter regime near {\emph{singular Hopf bifurcation}}, chaotic mixed-mode oscillations (MMOs), featuring concatenation of small and large amplitude oscillations are observed as long-lasting transients  before the system approaches its asymptotic state. To analyze the dynamical cause that initiates a large amplitude oscillation in an MMO orbit, the model is reduced to a suitable normal form near the singular-Hopf  point.  The normal form possesses a separatrix surface that separates two different types of oscillations. A  large amplitude oscillation is initiated if a trajectory moves from the ``inner" to the ``outer side" of this surface. A set of conditions on the normal form variables are obtained to determine whether a trajectory would exhibit another cycle of MMO dynamics  before experiencing a regime shift (i.e. approaching its asymptotic state). These conditions can serve as early warning signs for a sudden population shift in an ecosystem. 
\end{abstract}

\maketitle


\vspace{0.15in}



\section{Introduction}

Long-lasting transient dynamics play a crucial role in understanding complex ecosystems as they  can offer novel perspectives to elucidate mechanisms leading to regime shifts \cite{SC, Scheffer} in ecosystems that are under seemingly constant environments \cite{Hastings1,Hastings2}. Pest outbreaks, where dynamics shift significantly on a relative shorter timescale, or population abundance of small mammals such as voles in northern Sweden showing transition from large amplitude oscillations to  nearly steady-state dynamics, or transition in biomass of forage fishes from low density state to a much higher density state in the eastern Scotian Shelf ecosystem are some examples where regime shifts have occurred after long transients \cite{Hastings1,Hastings2}.

Typically, temporal variations in population densities of several species such as of larch budmoth in European subalpine forests, Pacific sardine and Northern Anchovy in Baja California etc. (see \cite{Baumgartner, EBFNL, Hastings2} and references therein) involve different timescales. 
 Using singular perturbation theory,  regular cycles of population outbreaks and collapses are modeled by relaxation oscillation cycles, commonly known as boom and bust cycles \cite{AS, LXY, MR, MR1}.  However, since  population cycles are not regular, generally not periodic in time, and feature transitions from one oscillatory state to another,  they can be perhaps better modeled by mixed-mode oscillations \cite{KC, Peet, SCT, Sadhunew}.   Mixed-mode oscillations (MMOs) are complex oscillatory patterns consisting of one or more small amplitude oscillations (SAOs) followed by large excursions of relaxation type, commonly known as large amplitude oscillations (LAOs) \cite{BKW, DGKKOW, K11}.  
 
 In this paper, we consider the model studied in \cite{Sadhunew} which unifies the two phenomena that appear in natural populations, namely MMOs featuring oscillations on different timescales, and their occurrence as long term transient dynamics on ecological timescales.  We primarily aim to understand the underlying mechanism responsible for transition from the transient state to the asymptotic state.
The model considered in \cite{Sadhunew} is a three-species predator-prey model, where two species of predators compete for their common prey with Holling type II functional response and explicit Lotka-Volterra type competition between the predators. Assuming the ratios between the  birth rates of the predators to the prey are extremely small, separation of timescales $\zeta$, is introduced into the system as a singular parameter \cite{MR, MR1}. This gives rise to a singularly perturbed system of equations with one-fast and two-slow variables. The model exhibits a variety of rich and complex dynamics with different types of irregular oscillations. In this paper, the main focus is to rigorously analyze the transient dynamics in a parameter regime near the {\emph{singular-Hopf point}} \cite{BB}, also referred to as {\emph{folded saddle-node of type II}} (FSN II) bifurcation point \cite{KPK, KW}.

In \cite{Sadhunew}, it was observed that when  both species of  predators are assumed to have similar predation efficiencies \cite{BD1, DH} near the  singular-Hopf regime, the coexistence equilibrium undergoes a supercritical Hopf bifurcation and chaotic MMOs appear as transient dynamics.  The transient dynamics could last for very long before the system approaches it asymptotic state (see figure \ref{bistable_transient}), 
reflecting that a supercritical bifurcation may not be regarded as a ``safe bifurcation" \cite{Scheffer} on ecological timescales.  
    \begin{figure}[h!]     
  \centering 
\subfloat[Dynamics in the phase space.]{\includegraphics[width=7.5cm]{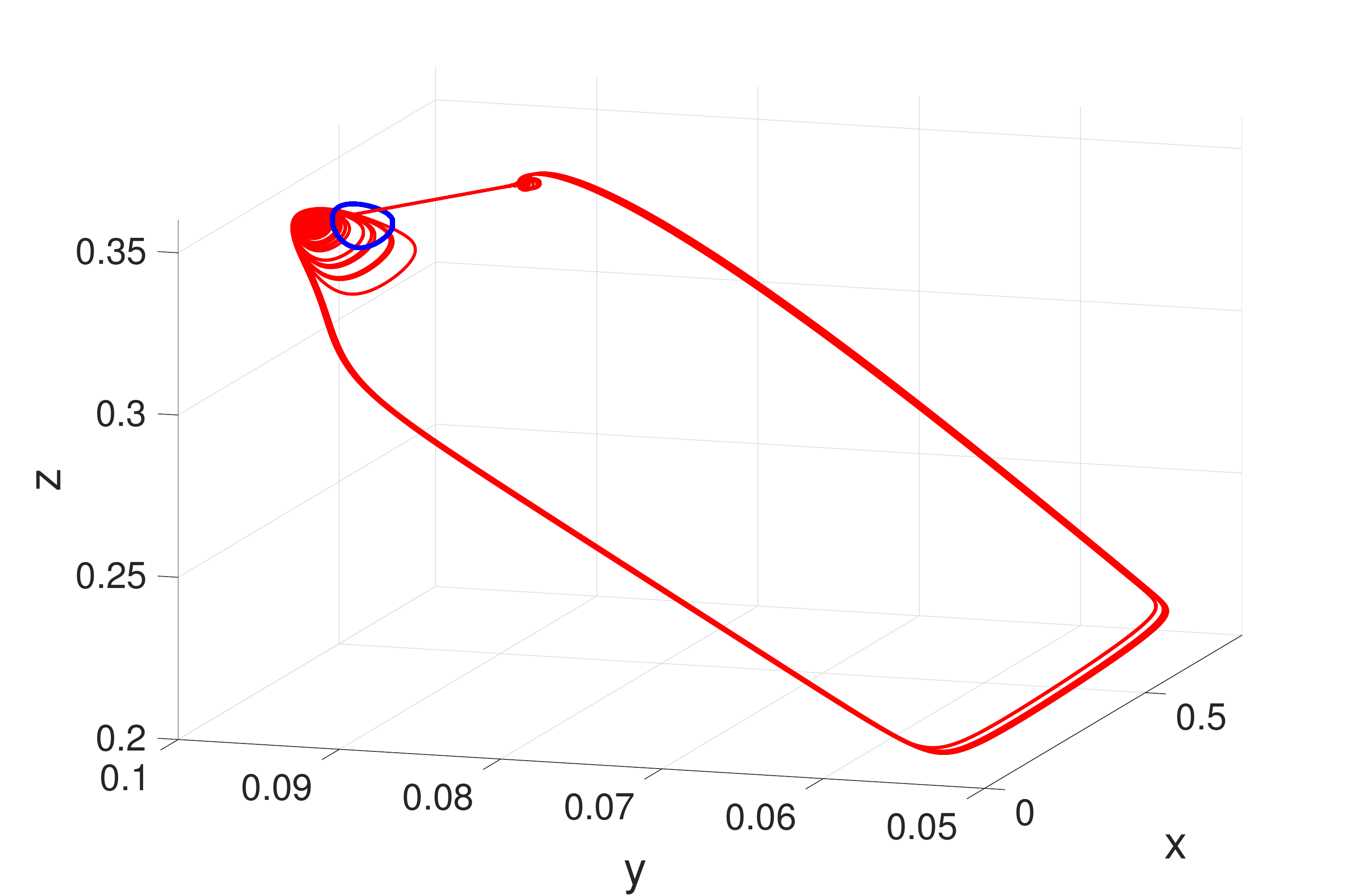}}\quad
 \subfloat[Time series of the state variable $x$.]{\includegraphics[width=7.5cm]{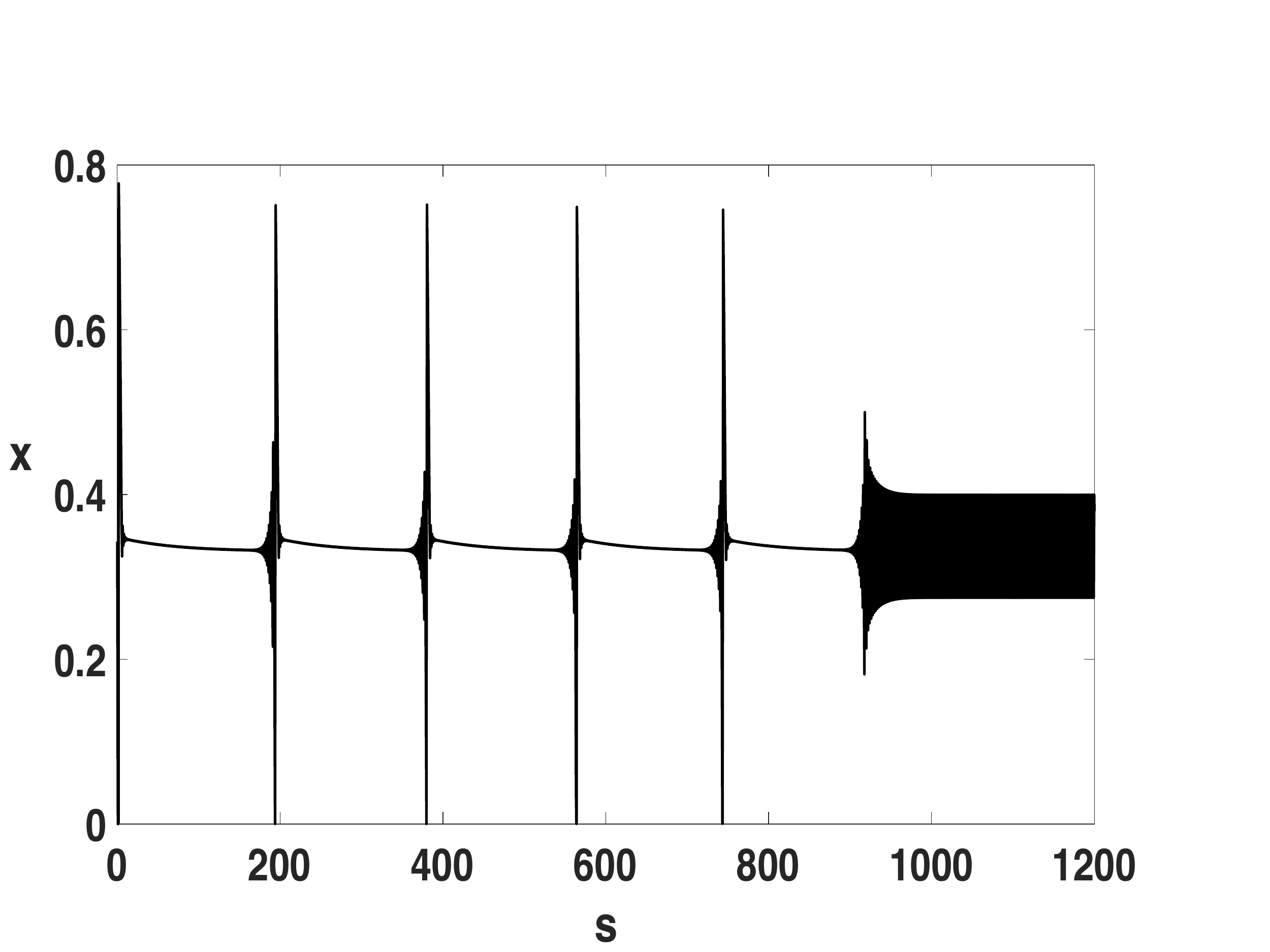}}\quad
  \caption{An illustration of regime shift in system (\ref{nondim3}). Long-term transient dynamics in form of mixed-mode oscillations eventually settle down to a small amplitude limit cycle (shown in blue),  born out of supercritical Hopf bifurcation. The parameter values chosen are $\zeta =0.01, \ \beta_1 =0.25,\  \beta_2 =0.35,\ c=0.4$,  $d=0.21,\ \alpha_{12} =0.5, \ \alpha_{21}=0.1$ and $h=0.783$.}
  \label{bistable_transient}
 \end{figure}
 The SAOs associated with the MMOs are organized by a slow passage through a {\emph{canard point}} and are further influenced by the local vector field around a saddle-focus equilibrium that lies in a vicinity of the canard point \cite{KW, SW}. As a trajectory spirals out along the unstable manifold of the equilibrium, it could either perform a large excursion  or approach the stable manifold of the periodic orbit.
It turns out that in either case, the local dynamics near the canard point and the equilibrium appear very similar, which makes any identification of early warning signs of a large population fluctuation extremely challenging (see figure \ref{transient_amplitude}).  
   \begin{figure}[h!]     
  \centering 
{\includegraphics[width=9.5cm]{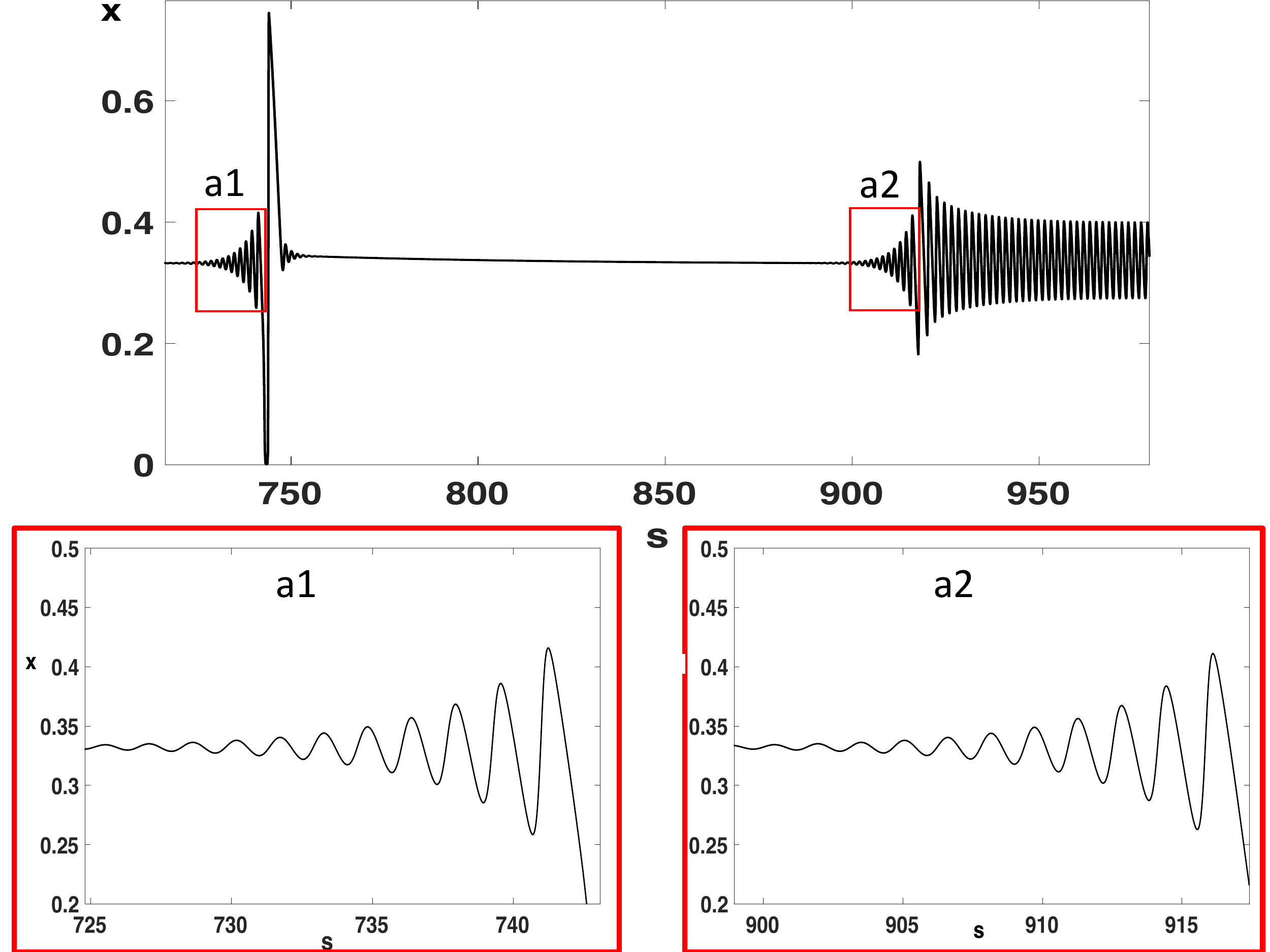}}
 \caption{A closer look of the time series in figure \ref{bistable_transient}  as the trajectory performs its  last large amplitude oscillation before approaching the small amplitude limit cycle. The insets show the similarity in local dynamics near a {\emph{canard point}} before the trajectory  executes a long excursion as shown in a1,  or before it approaches the asymptotic state 
 as shown in a2.}
 \label{transient_amplitude}
\end{figure}
 One of the goals of this paper is to analyze the phase space and provide a dynamical explanation for initiation of a large amplitude oscillation. The analysis is then employed to accurately predict the onset of a large change in the population density, and thus identify an early warning sign of an outbreak. Furthermore the analysis can be naturally extended to predict the timing of the transition from transient to the asymptotic dynamics.

To achieve the above goal, we take a systematic approach to analyze the transient dynamics near FSN II bifurcation. To do so, a suitable normal form reduction near the singular-Hopf point/FSN II point  is performed. With the aid of this  normal form, we prove that there exists a separatrix, a surface in the phase space that separates two different types of oscillations (see figure \ref{qsc}). In the singular limit, the separatrix approaches the surface of a parabolic cylinder that forms the integral surface of singular canard orbits of the normal form.  Trajectories starting above this surface make  a long excursion around the fold of the parabolic cylinder, resulting into a large amplitude oscillation,  whereas those starting below the surface exhibit SAOs.  As a system variable moves from one side of the separatrix to the other, a large amplitude oscillation is initiated, which will give us a precise mechanism for the generation of a large fluctuation. A similar approach was taken to characterize the excitability threshold in a class of planar neuronal models in \cite{DKR}, where the authors refer to this approach as the \emph{inflection method}. In a more recent work \cite{ADKR}, the method was extended to a canonical three-dimensional system with two slow variables that possesses a folded singularity (either a folded node or a folded saddle) and the slow dynamics is given by a constant slow drift. The analysis carried out in this paper 
pertains to the FSN II singularity scenario and offers a rigorous treatment of characterizing the local dynamics near such a point. To the best of our knowledge, this is the first ecological model involving two timescales, where a detailed analysis near the FSN II point is performed to investigate the long-lasting transients and develop a mechanism of identifying an early warning sign of a sudden population fluctuation.

The paper is organized as follows. In Section 2, we introduce the model and show some numerical results in a parameter regime near FSN II bifurcation, where long transients in form of chaotic MMOs are observed as a control parameter is varied.  In Section 3, a normal form reduction of the full system near FSN II  bifurcation point is performed, which is then analyzed to characterize the local dynamics near the equilibrium. Sufficient conditions on the normal form variables are obtained to determine whether a large amplitude oscillation will be initiated. The results are supported by  numerical simulations of the normal form in Section 4. Finally, we discuss our results and summarize our conclusions in the last section of the paper.


\section{Model Description and Numerical Results}

The ecological model considered in \cite{Sadhunew} in its dimensionless form reads as 

\begin{eqnarray}\label{nondim2}    \left\{
\begin{array}{ll} {x'}&= x\left(1-x-\frac{y}{\beta_1+x}-\frac{z}{\beta_2+x}\right)\\
 {y'}&=\zeta y\left(\frac{x}{\beta_1+x}-c-  \alpha_{12} z \right)\\
   {z'} &= \zeta  z\left(\frac{x}{\beta_2+x}-d -\alpha_{21} y-hz \right),
       \end{array} 
\right. \end{eqnarray}
where the primes denote differentiation with respect to the time variable $t$, and the variables $x$, $y$, $z$ measure the normalized prey abundance and the normalized abundance of the two predators respectively.  The parameter $\zeta$ measures the ratio of the growth rates of  predators to  prey and is assumed to satisfy  $0<\zeta \ll 1$. The dimensionless parameters $\beta_1$, $c$, $\alpha_{12}$ and $h$ respectively represent the predation efficiency \cite{BD1, DH}, rescaled mortality rate, rescaled interspecific and intraspecific competition coefficients of $y$. The parameters $\beta_2$, $d$ and $\alpha_{21}$ are analogously defined. We will assume that  $ 0< \beta_1, \beta_2, c,  d, \alpha_{12} , \alpha_{21} <1$ and $h>0$. The reader is referred to \cite{Sadhunew} for details.

 On rescaling $t$ by $\zeta$  and letting $s=\zeta t$, system $(\ref{nondim2})$ can be reformulated as

\begin{eqnarray}\label{nondim3}    \left\{
\begin{array}{ll} \zeta \dot{x} &= x\left(1-x-\frac{y}{\beta_1+x}-\frac{z}{\beta_2+x}\right) := x\phi(x,y,z)\\
    \dot{y}&=y\left(\frac{x}{\beta_1+x}-c-\alpha_{12} z\right) := y\chi(x,z)\\
    \dot{z}&=  z\left(\frac{x}{\beta_2+x}-d - \alpha_{21} y -hz  \right) := z\psi(x,y,z),
       \end{array} 
\right. 
\end{eqnarray}
where the overdot denotes differentiation with respect to the variable $s$ and $\phi=0$, $\chi=0$, and $\psi=0$ denote the nontrivial $x$, $y$, and $z$-nullclines respectively. The variables $t$ and $s$ are referred to as the fast and slow time variables respectively and the parameter $\zeta$ is regarded as the separation of timescales.

 The set of equilibria of the the fast system (\ref{nondim2}) in its singular limit forms the {\em{critical manifold}} $\mathcal{M}=T\cup S$, where
\bess T =\{ (0,y,z): y, z \geq 0\} \ \textnormal{and} \ S=\{(x,y,z) \in {\mathbb{R}^3}^+:\phi(x,y,z)=0\}.
\eess
The plane $T$ can be divided into two normally hyperbolic sheets $T^a=\{(0, y, z):\phi(0, y, z)<0\}$ and $T^r=\{(0, y, z): \phi(0, y, z)>0\}$ by the line $\mathcal{TC}=\{(0, y, z): y/\beta_1+z/\beta_2=1\}$, along which  transcritical bifurcations of the fast-subsystem occur. Similarly, the surface $S$  is divided into two normally hyperbolic sheets $S^a=S\cap \{ \phi_x(x, y, z)<0\}$ and $S^r=S\cap \{ \phi_x(x, y, z)>0\}$  by the curve $\mathcal{F} = S\cap \{ \phi_x(x, y, z)=0\}$, along which saddle-node bifurcations of the fast-subsystem occur.  In the singular limit of slow system (\ref{nondim3}), the corresponding flow called the reduced flow, restricted to $S$, has singularities along the fold curve $\mathcal{F}$, referred to as the folded singularities or canard points \cite{DGKKOW}.  These singularities are analyzed by desingularizing the reduced flow and are classified as folded nodes, folded saddles, folded foci or degenerate folded nodes based on the eigenvalues of the linearized matrix of the desingularized system (see \cite{Sadhunew} for details).  If the equilibrium of the full-system and a folded singularity merge together and split again, interchanging their type and stability, then a {\emph{folded saddle-node bifurcation of type II}} (FSN II)  occurs. This bifurcation corresponds to a transcritical bifurcation of the desingularized slow-subsystem, where the equilibrium crosses the fold curve $\mathcal{F}$. A Hopf bifurcation of system (\ref{nondim3}), referred to as {\emph{singular Hopf}} bifurcation \cite{BE, BB} occurs in $O(\zeta)$ neighborhood of  FSN II bifurcation.  Complex oscillatory dynamics such as mixed-mode oscillations (MMOs) can arise by a generalized folded node type canard phenomenon \cite{BKW} or singular Hopf bifurcation \cite{G}.

 Treating the intraspecific competition $h$ as the primary bifurcation parameter and the predation efficiency $\beta_1$ as the secondary parameter, system (\ref{nondim3}) was analyzed in \cite{Sadhunew}. One of the interesting  dynamics that is observed  is the existence of long lasting transients in form of MMOs past a supercritical Hopf bifurcation of the coexistence equilibrium point. The duration of the transient depends sensitively on the initial values of the state variables, and may last for a significantly long amount of time for certain initial conditions. Also, with the same initial condition, chaotic transients that last for different durations are observed as the control parameter is varied as shown in figure \ref{transient}. 
      \begin{figure}[h!]     
  \centering 
\subfloat[$h=0.784$.]{\includegraphics[width=7.5cm]{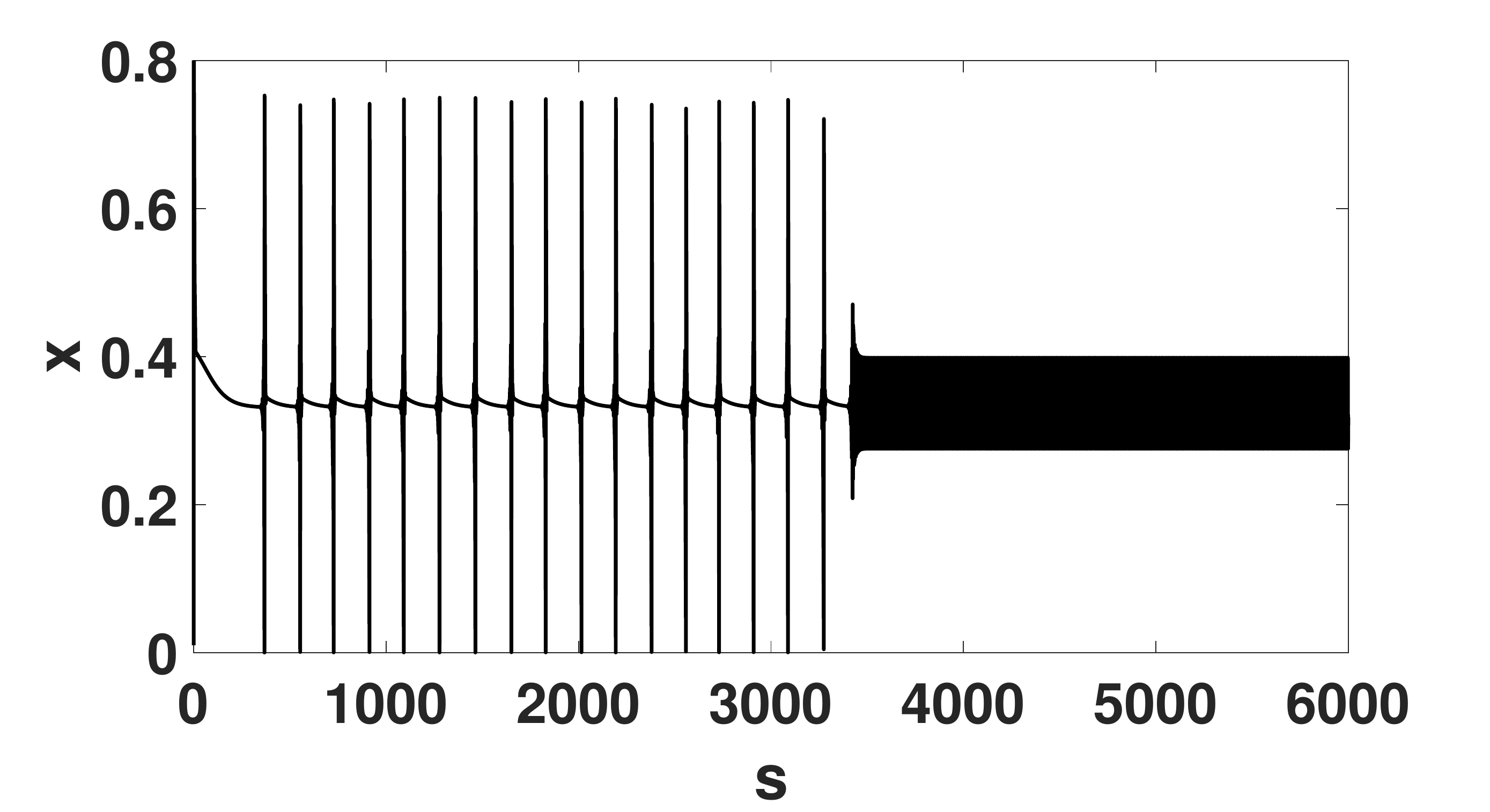}}\qquad
\subfloat[$h=0.786$.]{\includegraphics[width=7.5cm]{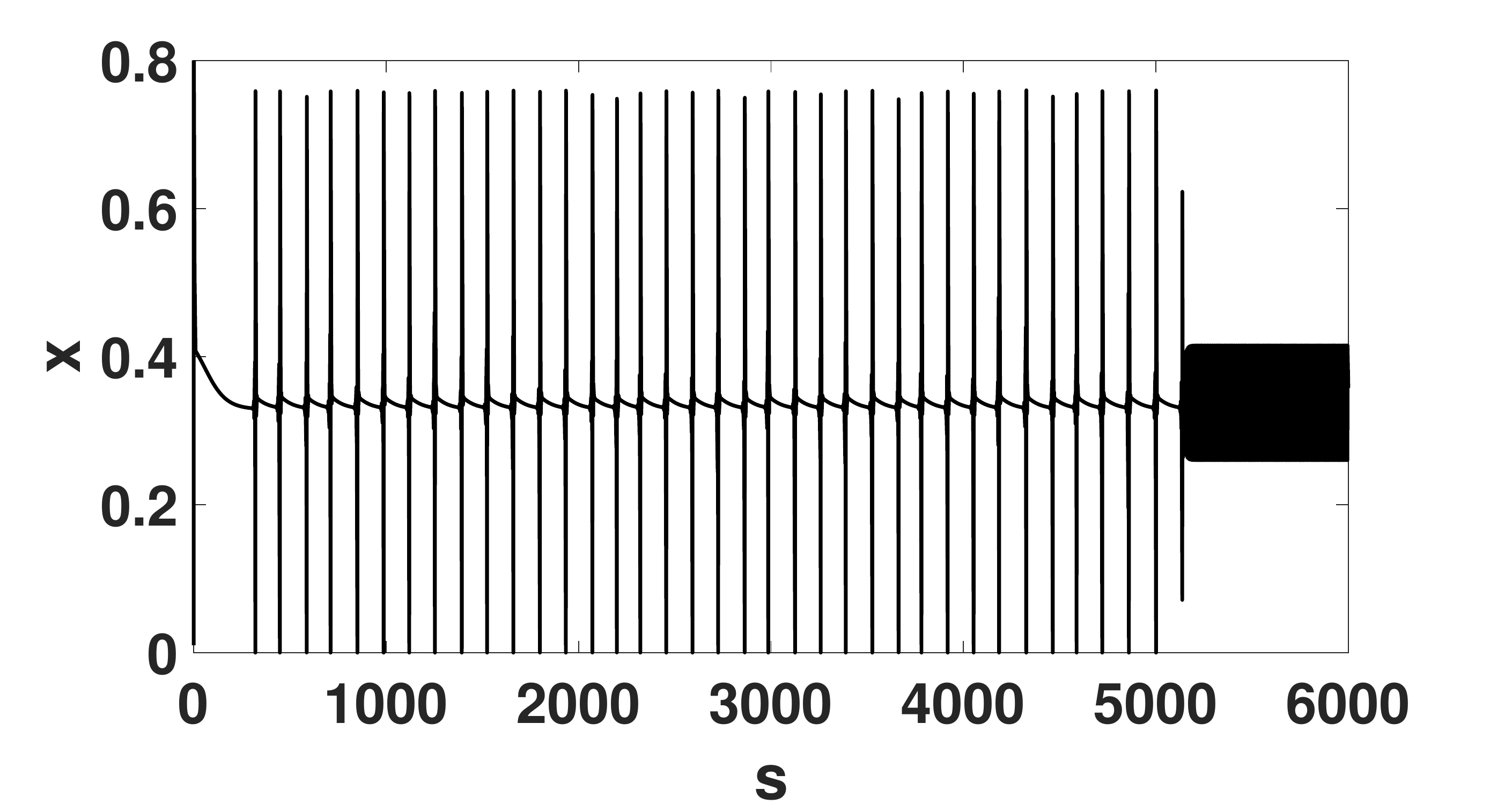}}
 \caption{Time series of the state variable $x$ of system (\ref{nondim3}) for different values of $h$. The initial data chosen is $(0.01, 0.01, 0.12)$. The parameter values chosen are $\zeta =0.01, \ \beta_1 =0.25,\  \beta_2 =0.35,\ c=0.4$,  $d=0.21,\ \alpha_{12} =0.5$ and  $\alpha_{21}=0.1$. Note that the duration of the transients depend sensitively on $h$.}
 \label{transient}
\end{figure}

  For the set of parameter values considered in figure \ref{transient}, a FSN II bifurcation of system (\ref{nondim3}) occurs at $h\approx 0.7785$  in the singular limit. The coexistence equilibrium $p_e$  is a stable spiral node at this parameter value. A supercritical Hopf bifurcation occurs at $h\approx 0.7803$ and a family of stable periodic orbits $\Gamma_h$ is born. Past the Hopf bifurcation, i.e. for $h >0.7803$,  $p_e$ is a saddle-focus with a two-dimensional unstable manifold and a one-dimensional stable manifold. The reduced system possesses a folded node singularity which helps in organizing the SAOs in an MMO cycle as a trajectory passes slowly through this point while approaching the equilibrium $p_e$.  This leads to very long epochs of SAOs (see figure \ref{bistable_transient}(A)).  Additional rotations are generated by the local vector around $p_e$.  The interaction between the repelling slow manifold with the repelling unstable manifold of $p_e$ causes the trajectory to jump, leading to a large amplitude oscillation.  On the other hand, if the trajectory gets trapped by the stable manifold of the periodic orbit $\Gamma_h$ while leaving the neighborhood of $p_e$, then it approaches $\Gamma_h$. In either case, whether the trajectory performs a large amplitude oscillation or approaches $\Gamma_h$,  the local dynamics near the canard point and the equilibrium are very similar as shown in the time series in  figure \ref{transient_amplitude}. This in turn makes any identification of early warning signs of  a sudden large fluctuation  extremely challenging.


\section{Normal form near the singular Hopf bifurcation}

In the previous section, we observed long transients  in form of chaotic MMOs  in a parameter regime near singular Hopf bifurcation. The equilibrium of the full system exists in a neighborhood of the fold curve $\mathcal{F}$. The intersections of the stable and unstable manifolds of the equilibrium with those of the slow manifolds can play roles in generating MMOs \cite{G}. Therefore it is desirable to study the behavior of a trajectory in a neighborhood of the unstable manifold of the equilibrium. This requires a detailed understanding of the  local dynamics around the equilibrium and investigating the phase space of system (\ref{nondim3}) near the fold curve. One way of achieving this goal would be  to numerically compute certain  (locally and globally) invariant manifolds and study the dynamics generated by their interaction \cite{DGKKOW, MKO}.  Computations of such manifolds are numerically challenging as they involve stiffness related issues.  Another approach is to reduce system (\ref{nondim3}) to a topologically locally equivalent system near the FSN II point on which a simpler geometric treatment  can be applied. In this paper, we take the latter approach and refer to the locally equivalent system as the normal form.

A normal form for singular Hopf bifurcation in one-fast and two-slow variables was first constructed by Braaksma \cite{BB}. Guckenheimer \cite{G}, later proposed another normal form for such systems which retains the form of the original system (namely one-fast and two slow-variables). However, to generate MMOs, a cubic term in the $x$-equation is added in \cite{G}. Since the equilibrium plays a central role in organizing the SAOs in the MMOs in system (\ref{nondim3}), the normal form in \cite{BB} is better suited to carry out the analysis. The construction involves a sequence of scalings and linear  transformations for Taylor expansion of (\ref{nondim3}) around the FSN II point, which lies $O(\zeta)$ from the equilibrium. We will adopt the transformations used in \cite{BB} to generate the normal form for system (\ref{nondim2}) near the FSN II point. 
The advantages of using this particular normal form are the following: (i) the location of the equilibrium does not change with the control parameter (ii) the stable manifold of the equilibrium is rectified along one of the coordinate axes (iii) besides covering the dynamics in a small neighborhood of FSN II, the normal form contains a cubic term which allows global returns of trajectories to the vicinity of the stable manifold of the equilibrium point, thereby retaining the original perspective
(iv) the equations for the scaled variables contain as many $O(1)$ terms as possible. It is also worth pointing out that  the time is scaled by a factor of $\sqrt{\zeta}$ to be consistent with the characteristic timescale for singular oscillations (the characteristic timescale for oscillations of the periodic orbits born due to Hopf bifurcation is $O(1/\sqrt{\zeta})$) in this normal form.  Furthermore, the numerical computations are much easier to perform as $\zeta \to 0$.

For  trajectories exhibiting similar local dynamics near the equilibrium before performing a large excursion in phase space or approaching the stable periodic orbit (see figure \ref{transient_amplitude}(B)), the normal form will be used to characterize their local behavior in a neighborhood of the equilibrium.  Furthermore, the sudden emergence of MMOs in an $O(\zeta)$ neighborhood of supercritical Hopf bifurcation will be also evident through the unfolding of the normal form. We remark that the local analysis performed near the FSN II singularity in this paper is akin to studying the corresponding blown-up vector field in the {\emph{central chart}} in the context of  geometric desingularization \cite{KS, KS1}. By considering appropriately defined {\emph{phase-directional charts}}, one can define a suitable global return map and prove existence of MMOs as in \cite{KPK}. However,  the present work does not focus on the global return mechanism and only pertains to studying the dynamics in a suitable neighborhood of the equilibrium.

The reduction of system (\ref{nondim3}) to its normal form  allows us to explicitly calculate Hopf bifurcation analytically. In the following we will treat $h$ as the bifurcating parameter. 
To this end, we rewrite system (\ref{nondim3})  as \begin{eqnarray}\label{normal1}    
\left\{
\begin{array}{ll} \zeta \dot{x} &=f_1(x, y, z,h)\\
    \dot{y}&=f_2(x, y, z,h) \\
    \dot{z}&= f_3(x,y, z,h),
       \end{array} 
\right. 
\end{eqnarray}
where $f_1(x, y, z, h)= x\phi(x,y,z,h), \ f_2(x, y, z,h) = y\chi(x,z,h)$ and $f_3(x, y, z,h) = z\psi(x,y,z,h)$. The overdot is with respect to the slow time variable $s$. Let $(\bar{x}, \bar{y}, \bar{z}, \bar{h})$ be a point where the following conditions hold:
\begin{itemize}

\item (P1) $\bar{\phi}=0,\  \bar{\chi}=0,\  \bar{\psi}=0$.

\item (P2) $\bar{\phi_x}=0$.

\item (P3) $\det J \neq 0$, where $J=  \begin{pmatrix}
(\bar{f}_1)_ x  & (\bar{f}_1)_y  &(\bar{f}_1)_z\\
(\bar{f}_2)_x & (\bar{f}_2)_y  & (\bar{f}_2)_z\\
(\bar{f}_3)_x & (\bar{f}_3)_y  & (\bar{f}_3)_z
 \end{pmatrix}$ 
 
 \item (P4)  $\begin{pmatrix}
 (\bar{f}_1)_y  & (\bar{f}_1)_z
 \end{pmatrix}\begin{pmatrix}
 (\bar{f}_2)_x \\
 (\bar{f}_3)_x 
 \end{pmatrix} <0.$ 
 
 \item (P5) $\bar{\phi}_{xx} \neq 0$.
 
 \item (P6) $-\begin{pmatrix} (\bar{f}_1)_{xx} & (\bar{f}_1)_{xy} & (\bar{f}_1)_{xz}
\end{pmatrix} J^{-1}
 \begin{pmatrix}
(\bar{f}_1)_h  \\
(\bar{f}_2)_h\\
(\bar{f}_3)_h
 \end{pmatrix} +(\bar{f}_1)_{xh}\neq 0$,

\end{itemize}
where bars denote the values of the expressions evaluated at $(\bar{x}, \bar{y}, \bar{z}, \bar{h})$. The conditions (P1) and (P2) indicate that a FSN II bifurcation of the reduced system corresponding to system (\ref{normal1}) occurs at $(\bar{x}, \bar{y}, \bar{z}, \bar{h})$, where  $(\bar{x}, \bar{y}, \bar{z})$ is a fold point.  Condition (P3) implies the existence of a smooth family of equilibria $(x_0(h), y_0(h), z_0(h))$ in a neighborhood of $\bar{h}$ via the implicit function theorem. Condition (P4) implies that the linearization of system (\ref{normal1})  at equilibria $(x_0(h), y_0(h), z_0(h))$ admits a pair of eigenvalues with singular imaginary parts for sufficiently small $\zeta$. Condition (P5) implies that the fold point is non-degenerate. Finally condition (P6) implies that $\frac{d\sigma}{dh}\neq 0$ at $(\bar{x}, \bar{y}, \bar{z},\bar{h})$, where $\sigma$ is the real part of the pair of eigenvalues with singular imaginary parts of the linearization of system (\ref{normal1}) at the equilibrium.

\begin{theorem}
\label{normal} Under conditions (P1)-(P6),  system (\ref{normal1}) can be written in the normal form:
\begin{eqnarray}\label{normal2}    
\left\{
\begin{array}{ll} \frac{du}{d\tau} &=v+\frac{u^2}{2}+\delta (\alpha(h) u +F_{13} uw +\frac{1}{6}F_{111}u^3)+O(\delta^2)\\
   \frac{dv}{d\tau} &=-u+O(\delta^2) \\
    \frac{dw}{d\tau} & = \delta (H_3 w +\frac{1}{2} H_{11} u^2) +O(\delta^2)
       \end{array} 
\right. 
\end{eqnarray}
with $\delta =O(\sqrt{\zeta})$ and $\tau=s/\delta$, 
where 
\begin{eqnarray}\label{del}    
\left\{
\begin{array}{ll}
\delta &= \frac{\sqrt{\zeta}}{{\omega}},\\
F_{13}&=  \frac{\bar{h}\bar{x}\bar{z}}{\beta_2+\bar{x}} +\bar{x} \big( \frac{\alpha_{12} \bar{y}}{\beta_1+\bar{x}}-\frac{\alpha_{21}\bar{z}(\beta_1+\bar{x})}{(\beta_2+\bar{x})^2} \Big) +\frac{(\beta_2-\beta_1)\omega^2}{2(\beta_1+\bar{x})(\beta_2+\bar{x})^2 \Big(\frac{\bar{y} }{(\beta_1+\bar{x})^3}+ \frac{\bar{z}}{(\beta_2+\bar{x})^3}\Big)},\\
F_{111} &=\frac{1}{\omega^2}\Big(\frac{\alpha_{12}\beta_2 \bar{x}\bar{y}\bar{z}}{(\beta_1+\bar{x})(\beta_2+\bar{x})^2} +
\frac{\alpha_{21}\beta_2 \bar{x}\bar{y}\bar{z}}{(\beta_1+\bar{x})^2(\beta_2+\bar{x})} +\frac{h\beta_2 \bar{x}\bar{z}^2}{(\beta_2+\bar{x})^3} \Big) 
 -\frac{3\omega^2}{2\bar{x}^2}\frac{\frac{\beta_1 \bar{y}}{(\beta_1+\bar{x})^4}+\frac{\beta_2 \bar{z}}{(\beta_2+\bar{x})^4}}{\Big(\frac{\bar{y} }{(\beta_1+\bar{x})^3}+ \frac{\bar{z}}{(\beta_2+\bar{x})^3}\Big)} \\
  &- \frac{1}{2\bar{x}\Big(\frac{\bar{y} }{(\beta_1+\bar{x})^3}+ \frac{\bar{z}}{(\beta_2+\bar{x})^3}\Big)}\Big( \frac{2\beta_1 \bar{x}\bar{y}}{(\beta_1+\bar{x})^4}+\frac{2\beta_2 \bar{x}\bar{z}}{(\beta_2+\bar{x})^4}-\frac{\beta_1^2 \bar{y}}{(\beta_1+\bar{x})^4}-\frac{\beta_2^2 \bar{z}}{(\beta_2+\bar{x})^4}\Big),\\
 H_3 &= \frac{\alpha_{12}\beta_2\bar{x}\bar{y}\bar{z}}{\omega^2(\beta_1+\bar{x})(\beta_2+\bar{x})^2}+\Big(\frac{\alpha_{21}\bar{z}(\beta_1+\bar{x})}{\beta_2+\bar{x}}-\bar{h}\bar{z}\Big)\Big(1-\frac{\beta_2\bar{x}\bar{z}}{\omega^2 (\beta_2+\bar{x})^3}\Big),\\
 H_{11} &=\Big(1-\frac{\beta_2\bar{x}\bar{z}}{\omega^2 (\beta_2+\bar{x})^3}\Big) \Big[ \frac{\beta_2\bar{z}}{\bar{x}(\beta_2+\bar{x})^3 \Big(\frac{\bar{y} }{(\beta_1+\bar{x})^3}+ \frac{\bar{z}}{(\beta_2+\bar{x})^3}\Big)}-\frac{1}{\omega^2}\Big(\frac{\alpha_{21}\beta_1 \bar{y}\bar{z}}{(\beta_1+\bar{x})^2} +\frac{h\beta_2 \bar{z}^2}{(\beta_2+\bar{x})^3} \Big)\Big]\\
 &- \frac{\beta_2\bar{x}\bar{z}}{\omega^2 (\beta_1+\bar{x})(\beta_2+\bar{x})^2} \Big[ \frac{\beta_1\bar{y}}{\bar{x}(\beta_1+\bar{x})^3 \Big(\frac{\bar{y} }{(\beta_1+\bar{x})^3}+ \frac{\bar{z}}{(\beta_2+\bar{x})^3}\Big)}-\frac{\alpha_{12}\beta_2 \bar{y} \bar{z}}{\omega^2(\beta_2+\bar{x})^2} \Big],\\
  \alpha(h)&= \frac{\bar{x}\bar{z}(\beta_2+\bar{x})\Big[-2\alpha_{12}\Big(\frac{\bar{y}}{(\beta_1+\bar{x})^3}+\frac{z}{(\beta_2+\bar{x})^3}\Big) +\frac{\beta_1(\beta_1-\beta_2)}{(\beta_2+\bar{x})^2(\beta_1+\bar{x})^3)}\Big]}{\frac{\alpha_{12}\beta_2}{\beta_2+\bar{x}}+\frac{\alpha_{21}\beta_1}{\beta_2+\bar{x}}-\frac{\beta_1\bar{h}(\beta_2+\bar{x})}{(\beta_1+\bar{x})^2}}\Big(\frac{h-\bar{h}}{\zeta}\Big) \\
 &- \frac{1}{\omega^2}\Big(\frac{\alpha_{12}\beta_2 \bar{x}\bar{y}\bar{z}}{(\beta_1+\bar{x})(\beta_2+\bar{x})^2} +
\frac{\alpha_{21}\beta_2 \bar{x}\bar{y}\bar{z}}{(\beta_1+\bar{x})^2(\beta_2+\bar{x})} +\frac{\bar{h}\beta_2 \bar{x}\bar{z}^2}{(\beta_2+\bar{x})^3} \Big) 
      \end{array} 
\right. 
\end{eqnarray}
and
\bess
\omega=\sqrt{\frac{\beta_1 \bar{x} \bar{y} }{(\beta_1+\bar{x})^3}+ \frac{\beta_2\bar{x}\bar{z}}{(\beta_2+\bar{x})^3}}.
\eess
The quantity $H_3$ is invertible which allows a further reduction of system (\ref{normal2}) using center manifold theory. The normal form is valid for $(x, y, z, h)=(\bar{x}+O(\sqrt{\zeta}), \bar{y}+O(\zeta) , \bar{z}+O(\zeta), \bar{h}+O(\zeta))$.
\end{theorem}

We refer to the work of Braaksma in \cite{BB} for the detailed proof. In terms of the new variables, as long as  $u(\tau), v(\tau), w(\tau) =O(1)$,  the dynamics of (\ref{normal1}) is topologically equivalent to the normal form (\ref{normal2}).

Condition (P3) is related to proving invertibility of $H_3$. In fact, it turns out that (see \cite{BB} for details)
\bess
H_3 = \frac{1}{\omega^2} \det J,
\eess
which is nonzero by assumption (P3). 
Expressing (P3) in terms of $(\bar{x}, \bar{y}, \bar{z}, \bar{h})$, yields
\bess
\frac{\bar{x}\bar{y}\bar{z}}{(\beta_1+\bar{x})(\beta_2+\bar{x})}\Big(\frac{\alpha_{21}\beta_1}{\beta_1+\bar{x}} +  \frac{\alpha_{12}\beta_2}{\beta_2+\bar{x}}- \frac{\beta_1 \bar{h}(\beta_2+\bar{x})}{(\beta_1+\bar{x})^2} \Big) \neq 0,
\eess
where $\bar{x}$ is a solution to the equation
\bes
  \label{xbareq}
 \alpha_{12}(1-\beta_1-2\bar{x})(\beta_1+\bar{x})(\beta_2+\bar{x})^2+  (\beta_2-\beta_1)(1-c)\bar{x}+ (\beta_2-\beta_1)c \beta_1 =0
\ees
with \bess \frac{1-\max\{\beta_1, \beta_2\}}{2} < \bar{x} < \frac{1-\min\{\beta_1, \beta_2\}}{2}, \ \textnormal{provided} \  \beta_1 \neq \beta_2.
\eess
 Furthermore, $\bar{h}$ can be expressed in terms of $\bar{x}$, namely,
\bes
\bar{h}= \frac{(\beta_2-\beta_1)(\frac{\bar{x}}{\beta_2+\bar{x}} -d)+\alpha_{21}(1-\beta_2-2\bar{x})(\beta_1+\bar{x})^2}{(1-\beta_1-2\bar{x})(\beta_2+\bar{x})^2}.
\ees 
Note that for $\beta_1\neq \beta_2$, (\ref{xbareq}) is a polynomial of $4$th degree, not factorable, and so we cannot solve for $\bar{x}$ analytically, and hence computing $\bar{h}$ explicitly remains challenging. However, we can have explicit expressions  for certain special cases described in the Appendix.


\subsection{Hopf bifurcation analysis}

Since $H_3\neq 0$, the nonhyperbolic $(u, v)$ part and the hyperbolic $w$ part in system (\ref{normal2}) are linearly decoupled, and hence one can further perform a center manifold reduction (see Appendix).
Combining the center manifold derivation along with Theorem 1 results into the following theorem:

\begin{theorem} \label{hopfthm} Consider system (\ref{normal1}) satisfying the assumptions (P1)-(P6) of Theorem 1. Then system (\ref{normal1}) undergoes a Hopf bifurcation at $h=\bar{h}+\zeta A+O(\zeta^{3/2})$, where $A$ is the solution of the equation
\bess
&\frac{\bar{x}\bar{z}(\beta_2+\bar{x})\Big[-2\alpha_{12}\Big(\frac{\bar{y}}{(\beta_1+\bar{x})^3}+\frac{z}{(\beta_2+\bar{x})^3}\Big) +\frac{\beta_1(\beta_1-\beta_2)}{(\beta_2+\bar{x})^2(\beta_1+\bar{x})^3)}\Big]A}{\frac{\alpha_{12}\beta_2}{\beta_2+\bar{x}}+\frac{\alpha_{21}\beta_1}{\beta_2+\bar{x}}-\frac{\beta_1\bar{h}(\beta_2+\bar{x})}{(\beta_1+\bar{x})^2}} \\
 &=\frac{1}{\omega^2}\Big(\frac{\alpha_{12}\beta_2 \bar{x}\bar{y}\bar{z}}{(\beta_1+\bar{x})(\beta_2+\bar{x})^2} +
\frac{\alpha_{21}\beta_2 \bar{x}\bar{y}\bar{z}}{(\beta_1+\bar{x})^2(\beta_2+\bar{x})} +\frac{\bar{h}\beta_2 \bar{x}\bar{z}^2}{(\beta_2+\bar{x})^3} \Big) 
\eess
for sufficiently small $\zeta>0$.
The Hopf bifurcation is super(sub)critical if 
\bess
\frac{1}{2}F_{111}- \frac{F_{13}H_{11} }{H_3}<(>) 0.
\eess
\end{theorem}

 The eigenvalues of the variational matrix of (\ref{normal2}) at the equilibrium $p_e=(0, 0, 0)$ up to higher order terms are 
\bes
\label{eig1}
\rho_1 = \delta H_3,\ \rho_{2,3} = \frac{1}{2} \Big[\alpha \delta \pm \sqrt{\alpha^2 \delta^2-4}\Big].
\ees
 If $H_3< 0$, then the equilibrium is a stable node or a stable spiral for $\alpha<0$, while it is a saddle-focus with two-dimensional unstable and one-dimensional stable manifold for $0<\alpha<2/\delta$.
 For $0<\alpha<2/\delta$, the flow generated by (\ref{normal2}) linearized about the origin is given by 
\begin{eqnarray}\label{flow1}    
\left\{
\begin{array}{ll} u(\tau) &= e^{\frac{\alpha \delta \tau}{2}}\Big[ u_0 \cos (\vartheta \tau) +\frac{1}{\omega}\Big(v_0+\frac{\alpha \delta u_0}{2} \Big)  \sin(\vartheta \tau)\Big],\\
\label{flow} v(\tau) &=  e^{\frac{\alpha \delta \tau}{2}}\Big[ \frac{1}{\vartheta}\Big(u_0-\frac{\alpha \delta v_0}{2} \Big)  \sin (\vartheta \tau) + v_0 \cos(\vartheta \tau)\Big],\\
w(\tau) &= e^{\delta H_3 \tau} \Big(w_0+\frac{\delta H_{11}}{2} \Big(A+\frac{4\vartheta^2 u_0^2+(2v_0 +\alpha \delta u_0 )^2}{8 \vartheta^2 \delta (\alpha-H_3)}\Big)\Big)\\
 &+ \frac{\delta}{2} H_{11}e^{\alpha \delta \tau} \Big[A \cos(2\vartheta \tau)+B \sin(2\vartheta \tau) 
+\frac{4\vartheta^2 u_0^2+(2v_0 +\alpha \delta u_0 )^2}{8 \vartheta^2 \delta (\alpha-H_3)} \Big],
   \end{array} 
\right. 
\end{eqnarray}
where $\vartheta = \sqrt{1-\frac{\alpha^2 \delta^2}{4}}$,
\bess
A &=& \frac{\delta(\alpha-H_3) (4\vartheta^2 u_0^2-(2v_0 +\alpha \delta u_0 )^2)- 8 \vartheta^3 u_0(2v_0+\alpha \delta u_0) } {8\vartheta^2 (4\vartheta^2+\delta^2(\alpha-H_3)^2)}, \\
B &=& \frac{ 4\vartheta^2 u_0^2-(2v_0 +\alpha \delta u_0 )^2 +2\delta(\alpha-H_3)  u_0(2v_0+\alpha \delta u_0) } {4 \vartheta (4\vartheta^2+\delta^2(\alpha-H_3)^2)} 
\eess
and $(u(0), v(0), w(0))=(u_0, v_0, w_0)$. 

On the other hand,  if $H_3> 0$, then the equilibrium is an unstable node/spiral for $2/\delta>\alpha>0$, while it is a saddle-focus with two-dimensional stable and one-dimensional unstable manifold for $\alpha<0$.

 \subsection{Analysis of the normal form}Treating $\delta$ as the singular parameter, note that system (\ref{normal2}) has a $1$-dimensional critical manifold $CM=\{(0, 0, w):w \in \mathbb{R}\}$.  The reduced flow is governed by the equation $ dw/d\tau_s = H_3 w$, so that $w(\tau_s)=w_0e^{H_3 \tau_s}$, where $\tau_s=\delta \tau$ is the slow time and $w_0=w(0)$. The layer problem of system (\ref{normal2}) reads as
\begin{eqnarray}\label{nonpertb}    
\left\{
\begin{array}{ll} \frac{du}{d\tau} &=v+\frac{u^2}{2}\\
   \frac{dv}{d\tau} &=-u\\
    \frac{dw}{d\tau} &=0.
    \end{array} 
\right. 
\end{eqnarray}

For each fixed $w$, system (\ref{nonpertb}) possesses constants of motion given by
\bes \label{periodic}
(u^2+2v-2)e^v=-k \ \textnormal{for}\ -\infty<k\leq 2.
\ees
A family of closed orbits exist for $0<k<2$. The periodic orbits approach the fixed point $(0,0)$ as $k\to 2$ and grow in size as $k\to 0$. The level curve $k=0$ separates periodic orbits surrounding $(0,0)$ from orbits that get unbounded with $u\to \pm \infty$ in finite time and will be referred to as the {\emph{singular canard solution}} as in \cite{KPK}. The unbounded orbits lie above the parabola $v=1-u^2/2$. The curve $k=0$ will be denoted by $\Gamma^0$. For $\delta>0$ sufficiently small, system (\ref{normal2}) can be viewed as a perturbation of (\ref{nonpertb}) and its dynamics are typically referred to as ``near-integrable" \cite{KPK}.

  \subsubsection{Parametrization of the slow variable in system (\ref{normal2})}
 Since $w$ evolves slowly in system (\ref{normal2}),  it can be replaced by a  parameter $\lambda$, so that the fast variables $(u, v)$ are governed by 
\begin{eqnarray}\label{normal_par}    
\left\{
\begin{array}{ll} \frac{du}{d\tau} &= \delta (\alpha + F_{13} \lambda) u+v+\frac{u^2}{2}+ \frac{\delta}{6}F_{111}u^3 \\
   \frac{dv}{d\tau} &=-u
       \end{array} 
\right. 
\end{eqnarray}
up to $O(\delta^2)$. Linearization of (\ref{normal_par}) yields that the eigenvalues at the origin are 
\bes  \label{eigenval} \sigma_{1,2}(\lambda) = \frac{1}{2} \Big(\delta (\alpha +F_{13} \lambda) \pm \sqrt{\delta^2 (\alpha +F_{13} \lambda)^2-4} \Big) +O(\delta^2).
\ees
Assuming $F_{13}\neq 0$, these eigenvalues are complex conjugates for 
\bess |\lambda|<\frac{2-\alpha \delta }{\delta |F_{13}|} +O(\delta),
\eess
  and real otherwise. A Hopf bifurcation occurs at $\lambda=\lambda_H(\alpha)= -\alpha/F_{13} +O(\delta)$, and the sign of $F_{111}$ determines the nature of bifurcation  ($F_{111}<0$ implies a supercritical Hopf).  
  
   \subsubsection{Analysis of system (\ref{normal_par}) when $F_{111}<0$ and $F_{13}>0$.}  Note that up to $O(\delta^2)$, the $u$-nullcline of (\ref{normal_par}) defined by 
  \bess U(\alpha, \lambda) :=\Big\{(u, v): v=-\Big(\delta (\alpha + F_{13} \lambda)+\frac{u}{2}+ \frac{\delta}{6}F_{111}{u^2}\Big) u \Big\}\eess
   is $S$-shaped with fold points $( u_f^{\pm},  v_f^{\pm})$, where
  \bess
 u_f^{\pm}(\alpha,\lambda) &=& \frac{1}{\delta F_{111}}\Big(-1\pm \sqrt{1-2\delta^2 F_{111}(\alpha+F_{13}\lambda)}\Big),\\
  v_f^{\pm}(\alpha,\lambda) &=& -\Big(\delta (\alpha + F_{13} \lambda)+\frac{u_f^{\pm}}{2}+ \frac{\delta}{6}F_{111}{(u_f^{\pm})^2}\Big) u_f^{\pm}.
  \eess
  
For $\lambda>\lambda_H(\alpha)$, $U$ has two ``attracting'' branches: $U^{+}_a =U\cap \{u<u_f^+\}$,  $U^{-}_a =U\cap \{u>u_f^-\}$, and a ``repelling'' branch: $U_r =U\cap \{u_f^+<u<u_f^-\}$. Orbits above the curve $\Gamma^0$ are no longer unbounded, but  typically make large excursions in the phase plane while closely following  $U^{\pm}_a$ as they transition from one branch to the other.  For each $\alpha$, as $\lambda$ passes through $\lambda_H(\alpha)$,  the left fold point $(u_f^+, v_f^{+})$ touches the $v$-nullcline at the origin.  At $\lambda= \lambda_H(\alpha)$, $(u_f^+, v_f^{+})$ is a canard point, and the other fold $(u_f^-, v_f^{-})$ is a jump point. For each fixed $\alpha>0$,  in the unfolding of the configuration of the $u$-nullcline, transition from Hopf solutions (SAOs)  to relaxation oscillations occur upon variation of $\lambda$, where $\lambda_H(\alpha)<\lambda\leq 0$ as shown in figure \ref{sao_lao_normal_fast_1}.  The existence and profiles of limit cycles of (\ref{normal_par}) with varying $\lambda$ can be obtained using the approach in \cite{KS}. We will denote the limit cycle of (\ref{normal_par}) by $\Gamma_{\alpha}^{\lambda}$.  Similar dynamics occur when $\alpha<0$ and $\lambda>\lambda_H(\alpha)>0$. 

 \begin{figure}[h!]     
  \centering 
  \subfloat[$\lambda=-0.74$]
  {\includegraphics[width=5.5cm]{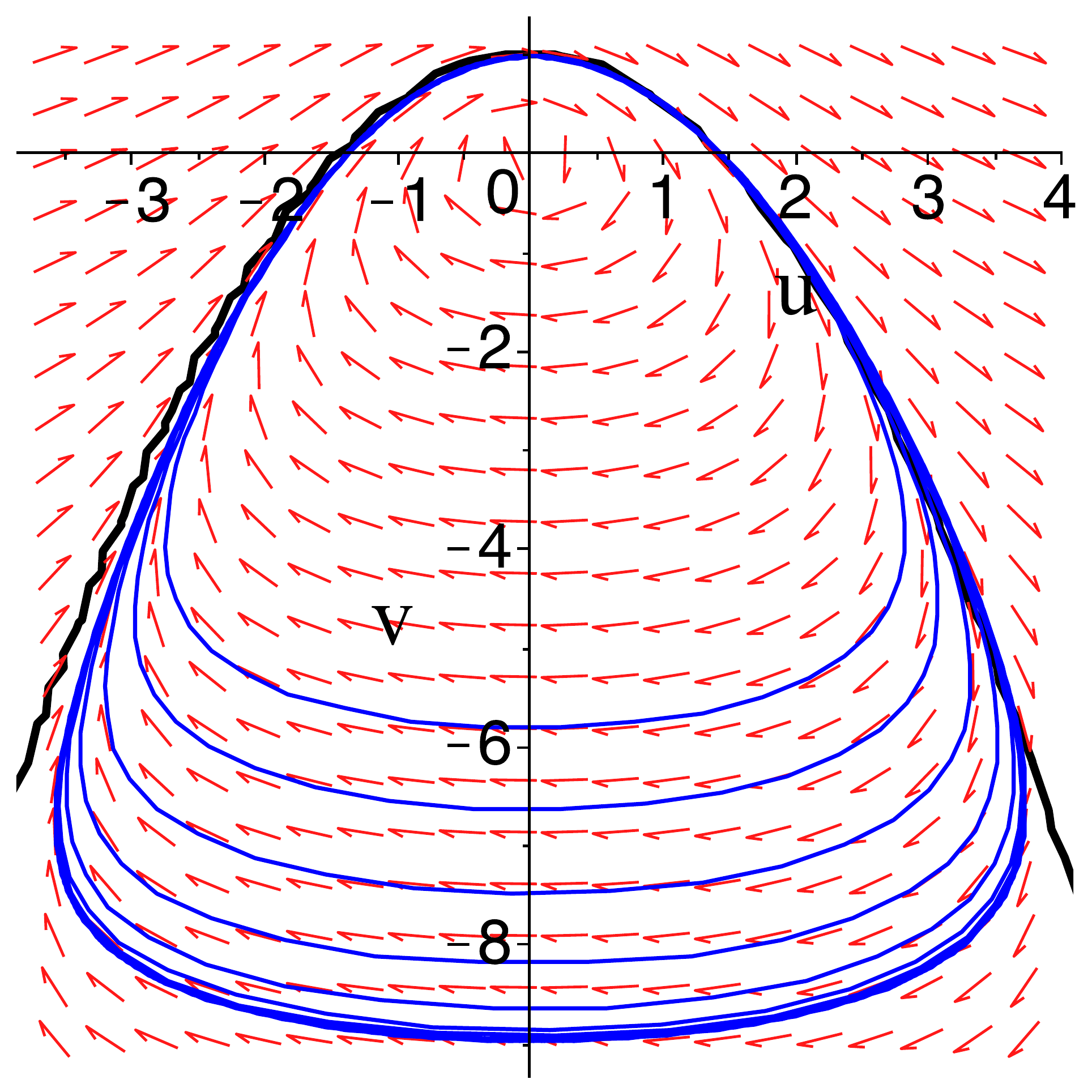}}\hspace{2.5cm}
  \subfloat[$\lambda=-0.7$]{\includegraphics[width=5.5cm]{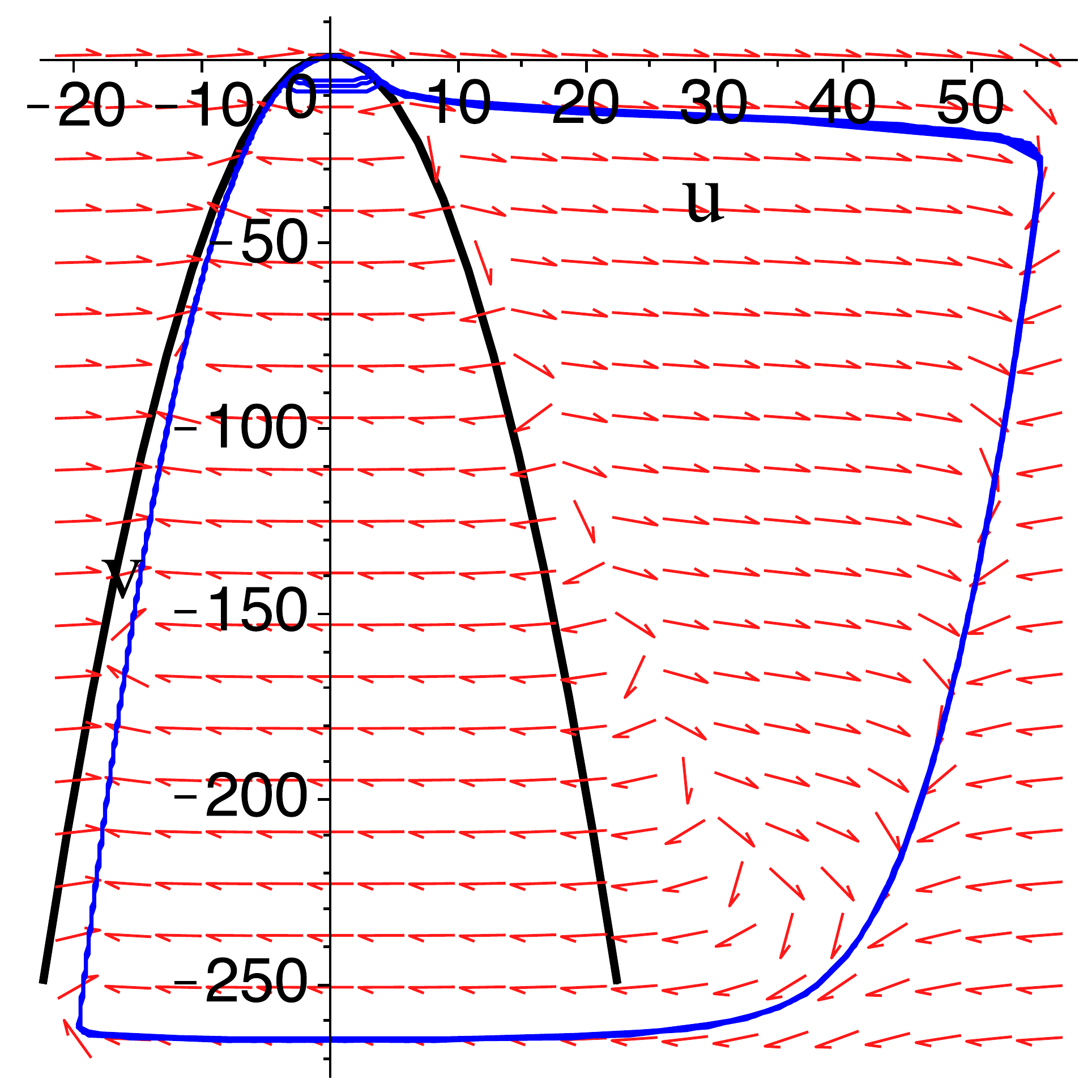}} 
  \caption{Behavior of system (\ref{normal_par}) as $\lambda$ is varied with $\alpha=0.4613$,  $\delta=0.078,$ $F_{13} = 0.1645$ and $F_{111} = -0.6833$.  Both trajectories start at $(1.5, -0.125)$ which  lie on the singular canard solution, $\Gamma^0$ shown in black.  (A) The trajectory settles down to the limit cycle exhibiting SAOs. (B)  The trajectory first performs SAOs and then settles down to the limit cycle exhibiting LAOs. Note the difference in scales between (A) and (B).}
  \label{sao_lao_normal_fast_1}
\end{figure}

\subsection{Detecting early warning signs using normal form variables} The next few subsections will focus on
 finding sufficient conditions on the normal form variables that will determine whether a trajectory would exhibit another cycle of MMO dynamics before approaching its asymptotic state. We will assume $F_{13}>0$, $F_{111}<0$, $H_{11}<0$ and $H_3<0$. Under these assumptions, a stable limit cycle $\Gamma_{\alpha}$ of system (\ref{normal2}) emerges through a supercritical Hopf bifurcation at $\alpha=0$. We assume that $0<\alpha\leq \alpha_s$ for some $\alpha_s>0$ such that $\Gamma_{\alpha}$ persists as the unique asymptotic attractor. We outline the analysis carried out in the next few subsections below.
 
In subsection 3.3.1, we consider the parameterized system (\ref{normal_par}) and compute its set of inflection points, i.e. a loci of phase space points along which trajectories of (\ref{normal_par}) have zero curvature \cite{DKR}. The trajectory that is tangential to the inflection set gives rise to the separatrix of the planar system (see figure \ref{inflection}).  We consider the first intersection  of the separatrix with the section $ \Delta:=\{u=0, v>0\}$ in negative time and study its behavior with varying $\lambda$.  This point will be denoted by ${v^{\star}_{\alpha}}(\lambda, \delta)$. For a trajectory originating in the region below the inflection set (i.e. with a negative curvature), a large amplitude oscillation is initiated if it returns to $\Delta$ above ${v^{\star}_{\alpha}}(\lambda, \delta)$.  We next study the relative positions of the limit cycles $\Gamma_{\alpha}^{\lambda}$ of (\ref{normal_par}) with respect to the separatrix in subsection 3.3.2 and discuss the existence of the critical value $\lambda= \lambda_c(\alpha)$ at which the separatrix  forms a closed orbit (see figure \ref{separatrix_limitcycle}). With the aid of this, we show that the parameter space $(\alpha, \lambda)$ can be divided into three dynamical regimes (see figure \ref{canard_explosion}) and large amplitude oscillations persist only if $\lambda>\lambda_c(\alpha)$. In subsection 3.3.3, we find an analytical approximation of $\lambda_c(\alpha)$ and obtain bounds on ${v^{\star}_{\alpha}}(\lambda)$. Finally, taking the evolution of $w$ into account, we consider the full system (\ref{normal2}) in subsection 3.3.4. Since $w = \lambda O(e^{\delta \tau})$ to its leading order, as long as the trajectory makes small oscillations,  the dynamics of the fast variables of (\ref{normal2}) can be governed by system (\ref{normal_par}), and thereby we argue that  a large amplitude oscillation is initiated in system (\ref{normal2}) if and only if the trajectory returns to $\Delta$ above ${v^{\star}_{\alpha}}(\lambda, \delta)$ for some $w(\tau)> \lambda_c(\alpha)$. The separatrix of the planar system is thus extended to a three dimensional surface, yielding a set of necessary and sufficient conditions that will ensure a large amplitude oscillation in system (\ref{normal2}).

 \subsubsection{Existence and properties  of a separatrix in system (\ref{normal_par})} Let  $(u(\tau), v(\tau))$ be a solution of (\ref{normal_par}) with $(u(0), v(0))=(0, v_{\lambda})$, $v_{\lambda}>0$ and $\lambda_H(\alpha)<\lambda\leq 0$. Note that  $u'(0)>0$ and $v'(\tau)<0$ as long as $u(\tau)>0$. Let $\tilde{\tau}_1(\lambda)$ be the first time such that $u'(\tilde{\tau}_1(\lambda))=0$. Then the solution can be expressed as $v=q(u)$ on the interval $[0, u(\tilde{\tau}_1(\lambda)))$ for some function $q(u)$, where $q(0)=v_{\lambda}$ and 
\bes \label{dervect} q'(u) = \frac{-u}{q(u)+\frac{u^2}{2}+\delta u \Big(\alpha + F_{13} \lambda)+ \frac{\delta}{6}F_{111}{u^2}\Big)} <0\  \textnormal{on} \  (0, u(\tilde{\tau}_1(\lambda))).\ees
   Also, let  $\tilde{\tau}(\lambda)>\tilde{\tau}_1(\lambda)$ be the first return time of the solution  $(u(\tau), v(\tau))$ that starts at $(0, v_{\lambda})$ to the positive $v$-axis. We will classify the solutions of (\ref{normal_par})  into two types: Type I and Type II. We will say that a solution exhibits a {\emph{Type I oscillation}}  on $[0, \tilde{\tau}(\lambda)]$  if $q'(u)$ is monotonically decreasing on $(0, u(\tilde{\tau}_1(\lambda)))$, and a {\emph{Type II oscillation}}  if $q'(u)$  has two local extrema in $(0, u(\tilde{\tau}_1(\lambda)))$.   
We will define the critical threshold $ {v^{\star}_{\alpha}}(\lambda, \delta)$  by 
   \bes
 \label{threshold} {v^{\star}_{\alpha}}(\lambda, \delta)= \sup \{ v_{\lambda}>0: q''(u)<0  ~\textnormal{on}~ (0, u(\tilde{\tau}_1(\lambda)))\},
\ees
which has the property that the corresponding solution $q(u)$ with $q(0)= {v^{\star}_{\alpha}}(\lambda, \delta)$ satisfies $q''(u^*_0)=0$ for some $u^*_0\in (0, u(\tilde{\tau}_1(\lambda)))$.  This solution  considered on the interval  $[0, \tilde{\tau}(\lambda)]$  will be referred to as the {\emph{separatrix}} $S_{\alpha, \lambda}(\delta)$, as it separates Type I oscillations  from Type II as shown in figure \ref{separatrix}.  It is to be noted that the SAOs observed in system (\ref{normal_par}) are of Type I while LAOs are of Type II (see figure \ref{sao_lao_normal_fast_1}). In the following, we will suppress the dependence of ${v^{\star}_{\alpha}}(\lambda, \delta)$ and $S_{\alpha, \lambda}(\delta)$ on $\delta$ and denote them by ${v^{\star}_{\alpha}}(\lambda)$ and $S_{\alpha, \lambda}$ respectively.

     \begin{figure}[h!]     
  \centering 
  {\includegraphics[width=8.5cm]{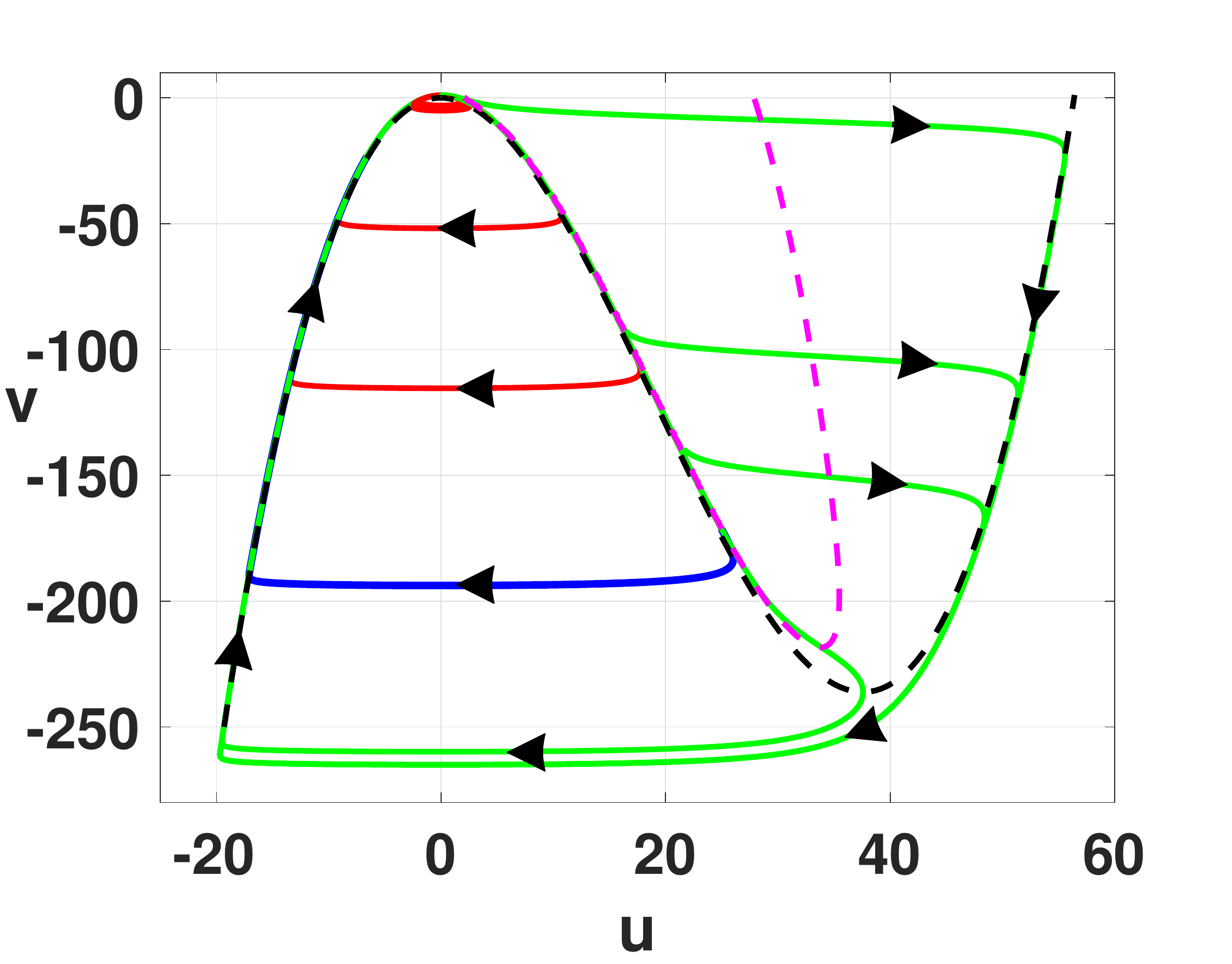}}
  \caption{Orbits of system (\ref{normal_par}) for parameter values $\delta=0.078,$ $F_{13} = 0.1645$, $F_{111} = -0.6833$, $\alpha=0.4613$ and $\lambda=-0.5$. The black dashed curve is the $u$-nullcline, $v=-f_{\lambda,\delta}(u)$, the dashed magenta curve is the curve of inflection points $g(u,v , \lambda)=0$, and the {\emph{separatrix}} is the blue curve separating Type I solutions (red trajectories) from Type II solutions (green trajectories).}
  \label{separatrix}
\end{figure}

We next find the inflection set of system (\ref{normal_par}) for each $\delta\geq 0$. This set is crucial in understanding the position of $S_{\alpha, \lambda}$ in the phase plane. It turns out that the inflection set has multiple connected components. For our analysis, we will only consider the component that is in the region $U:=\{(u,v):v>-f_{\lambda,\delta}(u), u \geq 0\}$, where 
\bes \label{vf}  f_{\lambda,\delta}(u):= \Big(\delta (\alpha + F_{13} \lambda)+\frac{u}{2}+ \frac{\delta}{6}F_{111}{u^2}\Big) u.
\ees
To this end, we consider the slope of the vector field of system (\ref{normal_par}) which satisfies
\bes \label{vf1} \frac{dv }{du} =\frac{-u}{v+f_{\lambda,\delta}(u)}.
\ees

  Differentiating (\ref{vf1})  with respect to $u$ yields
\[\frac{d^2v }{du^2} =\frac{-(v+ f_{\lambda,\delta}(u))^ 2+u\Big(f'_{\lambda,\delta}(u)(v+f_{\lambda,\delta}(u)) - u\Big)}{(v+ f_{\lambda,\delta}(u))^3},\]  whose roots $v^{\lambda,\delta}_{\pm}(u)$ form the different branches of the multiple components of the inflection set, where 
\bes \label{roots} v^{\lambda, \delta}_{\pm}(u)=\frac{-2f_{\lambda,\delta}(u)+uf'_{\lambda, \delta}(u) \pm \sqrt{u^2(f'^2_{\lambda, \delta}(u)-4)}}{2}. \ees 
It is clear that these branches exist only if $f'^2_{\lambda, \delta}(u)\geq4$. We will restrict $v^{\lambda,\delta}_{\pm}(u)$ to $U$ (these branches exist if $f'_{\lambda,\delta}(u)\geq 2$ in $U$) and consider the curve $g_{\delta}(u,v, \lambda)=0$ formed by these branches, i.e.
\bes \label{inflc} g_{\delta}(u,v, \lambda):=-(v+ f_{\lambda,\delta}(u))^ 2+u\Big(f'_{\lambda,\delta}(u)(v+f_{\lambda,\delta}(u)) - u\Big) =0 \ees
 We will refer to this curve as the inflection curve.

Note that the solution $v=q(u)$ of  (\ref{normal_par}) with $q(0)= v_{\lambda}$ satisfies $q''(u) <0$  if $q(u)< v^{\lambda, \delta}_{-}(u)$ for all $u \in (0, u(\tilde{\tau}_1(\lambda)))$. On the other hand, if there exists some $u \in (0, u(\tilde{\tau}_1(\lambda)))$ such that $q(u)> v^{\lambda, \delta}_{-}(u)$, then it follows from the direction vector field of  (\ref{normal_par}) that $q(u)$ must also intersect $v^{\lambda, \delta}_{+}(u)$ for some $u \in (0, u(\tilde{\tau}_1(\lambda)))$. Consequently, $q''(u)$ changes its sign twice. Hence for $\delta>0$, by definition of ${v^{\star}_{\alpha}}(\lambda)$, Type I solutions lie in the region where $g(u,v, \lambda)<0$, Type II solutions intersect with $g(u,v, \lambda)=0$ twice and that the separatrix $S_{\alpha, \lambda}$  is tangential to the curve $g(u,v, \lambda)=0$ at a unique point $(u^*_0(\lambda),  v^*_0(\lambda))$ (see figures \ref{separatrix} - \ref{inflection}), which is explicitly computed in Lemma \ref{lm0}. The separatrix will be defined as the positive and negative time orbits through $(u^*_0,  v^*_0)$ till they intersect with the positive $v$-axis.

     \begin{figure}[h!]     
  \centering 
  {\includegraphics[width=8.5cm]{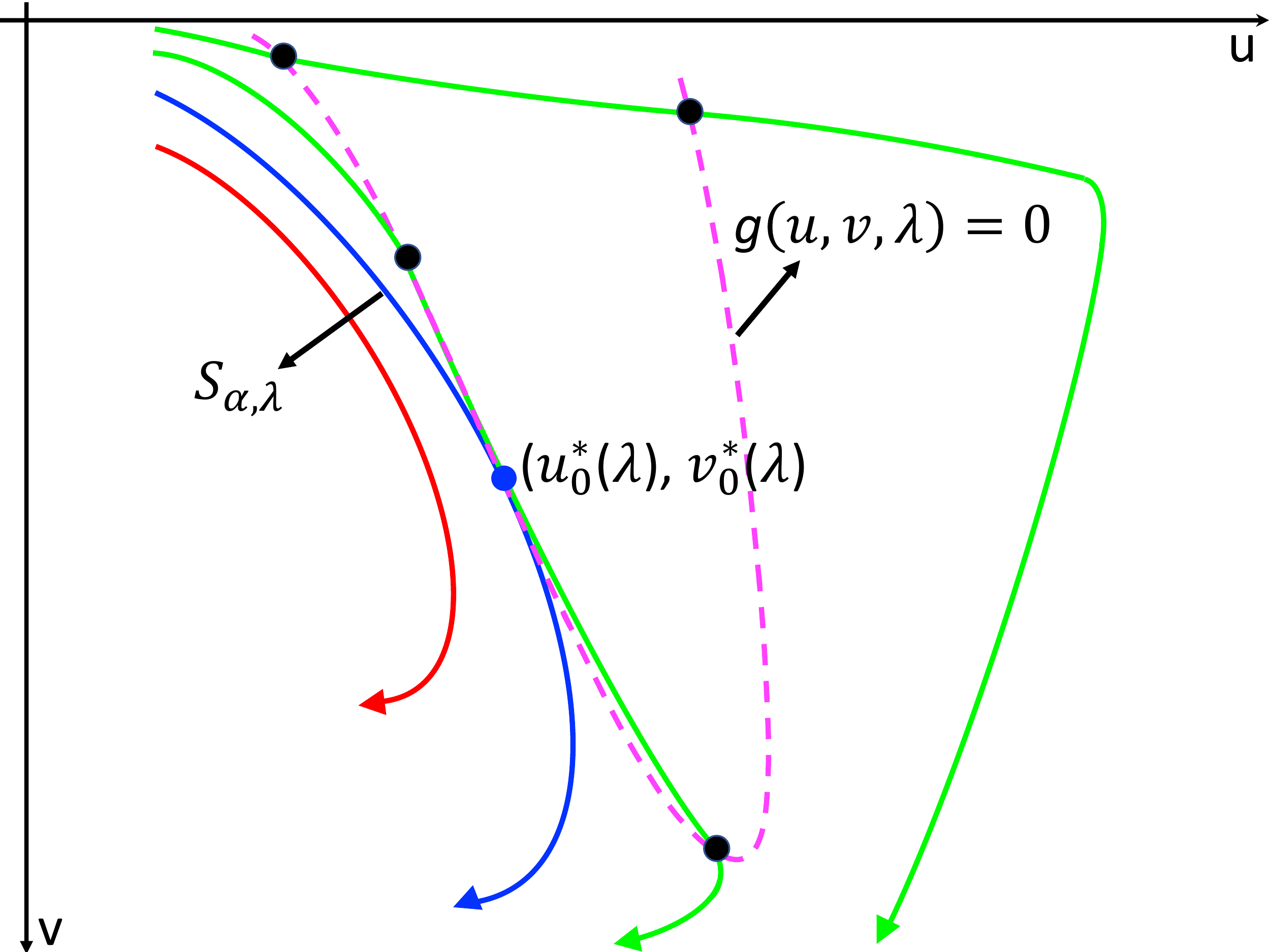}}
  \caption{A qualitative representation of orbits of system (\ref{normal_par}) near the curve of inflection points  $g(u,v, \lambda)=0$ for $\delta>0$. The separatrix $S_{\alpha,\lambda}$ (in blue) is tangential to $g(u,v , \lambda)=0$ at $(u^*_0(\lambda),  v^*_0(\lambda))$.  Orbits (in red) below $S_{\alpha,\lambda}$ do not intersect with $g(u,v, \lambda)=0$ and orbits (in green) above $S_{\alpha,\lambda}$  intersect with $g(u,v, \lambda)=0$ twice.  
  }
  \label{inflection}
\end{figure}

In the singular limit $\delta=0$, we note from (\ref{inflc}) that for each fixed $w$,  the inflection points of system (\ref{nonpertb}) are determined by solutions of the equation $u^4-4u^2-4v^2=0$. Based on the direction field of (\ref{nonpertb}), we have in this scenario that the inflection set is given by 
\bes \label{infsing} \left\{(u, v):v= -\frac{\sqrt{u^4 -4u^2}}{2}, u\geq 2\right\}.
\ees
Hence  it  follows from (\ref{periodic}) that a solution of (\ref{nonpertb}) with initial data $(0, v_{\lambda})$ intersects with (\ref{infsing}) if
 \bes \label{tempnew} u^2-2-\sqrt{u^4 -4u^2} = -k e^{\frac{\sqrt{u^4-4u^2}}{2}} \ \textnormal {for \ some\ } u>2. \ees
 Note that (\ref{tempnew}) has a solution if and only if $k<0$, which then implies that for each fixed $w=\lambda$, a solution  $v=q(u)$ of (\ref{nonpertb}) defined on the interval $[0, u(\tilde{\tau}_1(\lambda))]$ with  $q(0)= v_{\lambda}$, satisfies $q''(u)<0$ if and only if $0\leq k <2$. Consequently, we have from (\ref{threshold}) that ${v^{\star}_{\alpha}}(\lambda) =1$ for all $\lambda$ when $\delta=0$.

 We will now compute the point $(u^*_0(\lambda),  v^*_0(\lambda))$ where $S_{\alpha, \lambda}$  is tangential to the curve $g(u,v, \lambda)=0$.


\begin{lemma} \label{lm0}  Assume $F_{13}>0$, $F_{111}<0$ and $\alpha >0$. Then for each $\lambda$ satisfying $\lambda_H(\alpha)<\lambda \leq 0$ and sufficiently small $\delta>0$, the separatrix $S_{\alpha,\lambda}$ of system (\ref{normal_par}) is tangential to the curve $g(u,v, \lambda)=0$ at the point  $(u^*_0(\lambda),  v^*_0(\lambda))$, where  $u^*_0(\lambda) =-1/(\delta F_{111})$ and  $v^*_0 (\lambda)= v^{\lambda, \delta}_{-}(u^*_0(\lambda))$ with $v^{\lambda,\delta}_{-}(u)$  defined by (\ref{roots}). 
\end{lemma}

\begin{proof} Consider the interval $(u^{\lambda}_1, u^{\lambda}_2)$, where $f'_{\lambda,\delta}(u)>2$. Note  from the definition of $f_{\lambda,\delta}(u)$ that $u^{\lambda}_1 >0$ and $v^{\lambda,\delta}_{-}(u)$ defined by (\ref{roots}) exists on $(u^{\lambda}_1, u^{\lambda}_2)$. Differentiating $v^{\lambda,\delta}_{-}(u)$ yields
\bess \frac{dv^{\lambda,\delta}_{-}}{du} &=& \frac{1}{2} \Big(-f'_{\lambda,\delta}(u)+uf''_{\lambda,\delta}(u) -\frac{f'^2_{\lambda,\delta}(u)+uf'_{\lambda,\delta}(u)f''_{\lambda,\delta}(u) -4}{\sqrt{f'^2_{\lambda,\delta}(u)-4}}\Big)\\
&=& \frac{1}{2} \Big(-\delta\Big(\alpha+F_{13} \lambda -\frac{F_{111}}{2}u^2 \Big) -\frac{f'^2_{\lambda,\delta}(u)+uf'_{\lambda,\delta}(u)f''_{\lambda,\delta}(u) -4}{\sqrt{f'^2_{\lambda,\delta}(u)-4}}\Big).
\eess
 This implies that  $dv^{\lambda,\delta}_{-}/du<0$ at least  as long as $f''_{\lambda,\delta}(u)\geq 0$. 
 
Suppose that $v=q(u)$,  $v(0)=v_{\lambda}>0$ is a solution of system (\ref{normal_par}) with $q'(u)<0$ on the interval $(0, u(\tilde{\tau}_1(\lambda)))$ for some $\lambda$ satisfying $\lambda_H(\alpha)<\lambda \leq 0$.  Note that such a solution satisfies $g(0, v_{\lambda}, \lambda)=-v^2_{\lambda}<0$.   Hence for a transversal or a tangential intersection of $v=q(u)$ with $v^{\lambda,\delta}_{-}(u)$  on the interval $[0, u(\tilde{\tau}_1(\lambda)))$, we must have that $q'(u)\geq dv^{\lambda,\delta}_{-}/du$ at the point of contact. Let $(u^*_0,  v^*_0)$ be the tangential point of intersection. Then we must have  $v^*_0=v^{\lambda,\delta}_{-}(u^*_0)$ and
 \[ q'(u^*_0)= \frac{-u^*_0} {v^*_0+f_{\lambda,\delta}(u^*_0)}= \frac{dv^{\lambda,\delta}_{-}}{du}\Big|_{(u^*_0,  v^*_0)}.\] This occurs only if $f_{\lambda,\delta}''(u^*_0)=0$, which implies that  $u^*_0 =-1/(\delta F_{111})$.
\end{proof}

 Based on Lemma \ref{lm0} and definition (\ref{threshold}) of the critical threshold,  we note that ${v^{\star}_{\alpha}}(\lambda)$ is the first point of intersection of the negative time orbit through $(u^*_0,  v^*_0)$ with the positive $v$-axis. The separatrix $S_{\alpha, \lambda}(\delta)$ is the trajectory that starts at $(0, v_{\alpha}^*(\lambda))$ and ends at $(0, {v_{\alpha}^*(\lambda)}^r)$, where  ${v_{\alpha}^*(\lambda)}^r$ is the first return position of this trajectory on the positive $v$-axis.
It is also clear that $S_{\alpha, \lambda}(\delta)$ will form a closed orbit if ${v^{\star}_{\alpha}}(\lambda) = {v_{\alpha}^*(\lambda)}^r$. Hence,  a trajectory of  (\ref{normal_par}) can transition from Type I oscillations to Type II (or vice versa) if it leaves (enters) the domain bounded by the separatrix.
By continuous dependence of solutions on parameters, $S_{\alpha, \lambda}(\delta)$  must lie close to the singular canard orbit $\Gamma^0$ in the region $\{(u, v): u \in [\rho, \rho], v \in [1-\rho, 1+\rho]\}$ for any $\rho>0$ as $\delta\to 0$. Hence for all $\delta>0$ sufficiently small, there exist $\bar{\delta}>0$ and $M>0$ independent of $\delta$ and $\lambda$ such that $|{v^{\star}_{\alpha}}(\lambda) -1|\leq M  \delta$ for all $\lambda_H(\alpha)<\lambda \leq 0$ and $0<\delta<\bar{\delta}$.

We next study the behavior of ${v^{\star}_{\alpha}}(\lambda)$  and prove that it monotonically decreases with $\lambda$.

\begin{lemma}\label{lm1} Assume $F_{13}>0$, $F_{111}<0$ and $\alpha >0$. Then for each $\lambda$ satisfying $\lambda_H(\alpha)<\lambda \leq 0$ and sufficiently small $\delta>0$, the critical threshold ${v^{\star}_{\alpha}}(\lambda)$ monotonically decreases with $\lambda$.
\end{lemma}
\begin{proof} From Lemma \ref{lm0}, we have that $(u^*_0(\lambda),  v^*_0(\lambda))$ satisfies $g(u^*_0(\lambda), v_{-}(u^*_0(\lambda)),\lambda)=0$, 
where $u^*_0 =-1/(\delta F_{111})$.
Differentiating $g(u^*_0, v_{-}(u^*_0), \lambda)=0$ with respect to $\lambda$, we obtain that
\bess (-2 v^*_0-2f_{\lambda,\delta}(u^*_0)+u^*_0f'_{\lambda,\delta}(u^*_0)) \frac{\partial  v^*_0 }{\partial \lambda} -\Big(-2 \frac{\partial f_{\lambda,\delta}}{\partial \lambda} + u^*_0 \frac{\partial f'_{\lambda,\delta}}{\partial \lambda}\Big)v^*_0\\
+(-2f_{\lambda,\delta}(u^*_0)+u^*_0f'_{\lambda,\delta}(u^*_0))\frac{\partial f_{\lambda,\delta}}{\partial \lambda} 
+u^*_0f(u^*_0)\frac{\partial f'_{\lambda,\delta}}{\partial \lambda} &=&0,\\
\textnormal{i.e.} \ (-2 v^*_0-2f_{\lambda,\delta}(u^*_0)+u^*_0f'_{\lambda,\delta}(u^*_0)) \frac{\partial  v^*_0 }{\partial \lambda} - \delta u^*_0F_{13} (v^*_0 +f_{\lambda,\delta}(u^*_0)-u^*_0f'_{\lambda,\delta}(u^*_0))&=&0.
\eess
By the definition of $v^*_0$ and (\ref{roots}), it follows that 
\[-2 v^*_0-2f_{\lambda,\delta}(u^*_0)+u^*_0f'_{\lambda,\delta}(u^*_0))= \sqrt {{u^*_0}^2(f'^2_{\lambda,\delta}(u^*_0)-4)}>0\] and that 
\[ v^*_0 +f_{\lambda,\delta}(u^*_0)-u^*_0f'_{\lambda,\delta}(u^*_0)=-\frac{1}{2}u^*_0f'_{\lambda,\delta}(u^*_0)-\frac{1}{2}\sqrt {{u^*_0}^2(f'^2_{\lambda,\delta}(u^*_0)-4)}<0. \] 
Hence, 
\[\frac{\partial  v^*_0 }{\partial \lambda} = \frac{-\delta F_{13} u^*_0 \Big(u^*_0f''_{\lambda,\delta}(u^*_0)+\sqrt {u^*_0(f'^2_{\lambda,\delta}(u^*_0)-4)}\Big)}{2\sqrt {u^*_0(f'^2_{\lambda,\delta}(u^*_0)-4)}} <0.\]
 Denoting the negative time orbit connecting $(u^*_0(\lambda),  v^*_0(\lambda))$ with $(0, {v^{\star}_{\alpha}}(\lambda))$ by $v=q_{\lambda}(u)$, we will show that no two such orbits can intersect if $\lambda_1\neq \lambda_2$. Suppose on the contrary, let $q_{\lambda_1}(u)$ intersect with $q_{\lambda_2}(u)$ at some point $0<\tilde{u}<u^*_0$, where  $q_{\lambda_i}(u^*_0)= v^*_0(\lambda_i)$, $i=1,2$. Without loss of generality, suppose $\lambda_1>\lambda_2$. The  monotonic dependence of $v^*_0(\lambda)$ implies that $v^*_0(\lambda_1)<v^*_0(\lambda_2)$. Hence,  for the intersection to occur, we must have $q'_{\lambda_1}(\tilde{u})<q'_{\lambda_2}(\tilde{u})$. However, we note from (\ref{dervect}) that
\[ q'_{\lambda_1}(\tilde{u})-  q'_{\lambda_2}(\tilde{u})=\frac{\delta \tilde{u} F_{13}(\lambda_1-\lambda_2)}{(q_{\lambda_1}(\tilde{u})+f_{\lambda_1}(\tilde{u}))(q_{\lambda_2}(\tilde{u})+f_{\lambda_2}(\tilde{u}))}>0, \]   contradicting the existence of $\tilde{u}$. Consequently, $q_{\lambda_1}(u) < q_{\lambda_2}(u)$ for $u\in [0,u^*_0]$ and  thus the monotonic dependence of ${v^{\star}_{\alpha}}(\lambda)$  with respect to $\lambda$ follows. 
\end{proof}

\begin{remark}\label{rmk11} The proof of  Lemma \ref{lm1} can be adapted to show that  ${v^{\star}_{\alpha}}(\lambda)$ also monotonically decreases with $\alpha$.
\end{remark}

 \subsubsection{Relative position of limit cycles of system (\ref{normal_par}) with respect to the separatrix $S_{\alpha, \lambda}$}  For each $\alpha>0$ and  $\lambda \in( \lambda_H(\alpha), 0]$, suppose that the limit cycle $\Gamma_{\alpha}^{\lambda}$  of system (\ref{normal_par}) intersects with the positive $v$-axis at  $(0, v_{\alpha}^{\lambda})$.  It then follows from Lemma \ref{lm0}  that  $v_{\alpha}^{\lambda} < v_{\alpha}^*(\lambda)$ if $\Gamma_{\alpha}^{\lambda}$ exhibits Type I oscillations, and $v_{\alpha}^{\lambda}>  v_{\alpha}^*(\lambda)$ if $\Gamma_{\alpha}^{\lambda}$ exhibits Type II oscillations. It is also clear that for $\alpha$ sufficiently small,  $\Gamma_{\alpha}^{\lambda}$ is  below $\Gamma^0$, the singular canard orbit, for all $\lambda_H(\alpha)< \lambda \leq 0$, which implies that  $v_{\alpha}^{\lambda} < v_{\alpha}^*(\lambda)$ for all $\lambda$ in that range. Let $\alpha_c>0$ be defined by 
\bess
\alpha_c:=\sup \{0<\alpha \leq \alpha_s:  \textnormal{for all}~ \lambda \in( \lambda_H(\alpha), 0], (u(\tau), v(\tau)) \cap \Gamma^0 = \emptyset ~ \textnormal{for  all} ~\tau>0 \},\eess 
where $(u(\tau), v(\tau))$ is a solution of  (\ref{normal_par}) with  $0<|u(0)|, |v(0)| \ll 1$. Since $\Gamma_{\alpha}^{\lambda}$ grows with $\lambda$, while $v_{\alpha}^*(\lambda)$  decreases monotonically with $\lambda$ (follows from Lemma \ref{lm1}), there must exist a curve  $\lambda_c(\alpha)$ (that depends on $\delta$) defined for $\alpha>\alpha_c$ such that $v_{\alpha}^{\lambda_c} =  v_{\alpha}^*(\lambda_c)$. In other words, the separatrix $S_{\alpha,\lambda}$ forms a closed orbit at $\lambda=\lambda_c(\alpha)$.
 This in turn implies that for $\alpha>\alpha_c$, $\Gamma_{\alpha}^{\lambda}$ exhibits Type I oscillations if $\lambda<  \lambda_c(\alpha)$, while for  $\lambda> \lambda_c(\alpha)$,  $\Gamma_{\alpha}^{\lambda}$ performs  Type II oscillations as shown in figure \ref{separatrix_limitcycle}. 
  \begin{figure}[h!]     
  \centering 
  \subfloat[$\lambda< \lambda_c(\alpha)$]
  {\includegraphics[width=7.5cm]{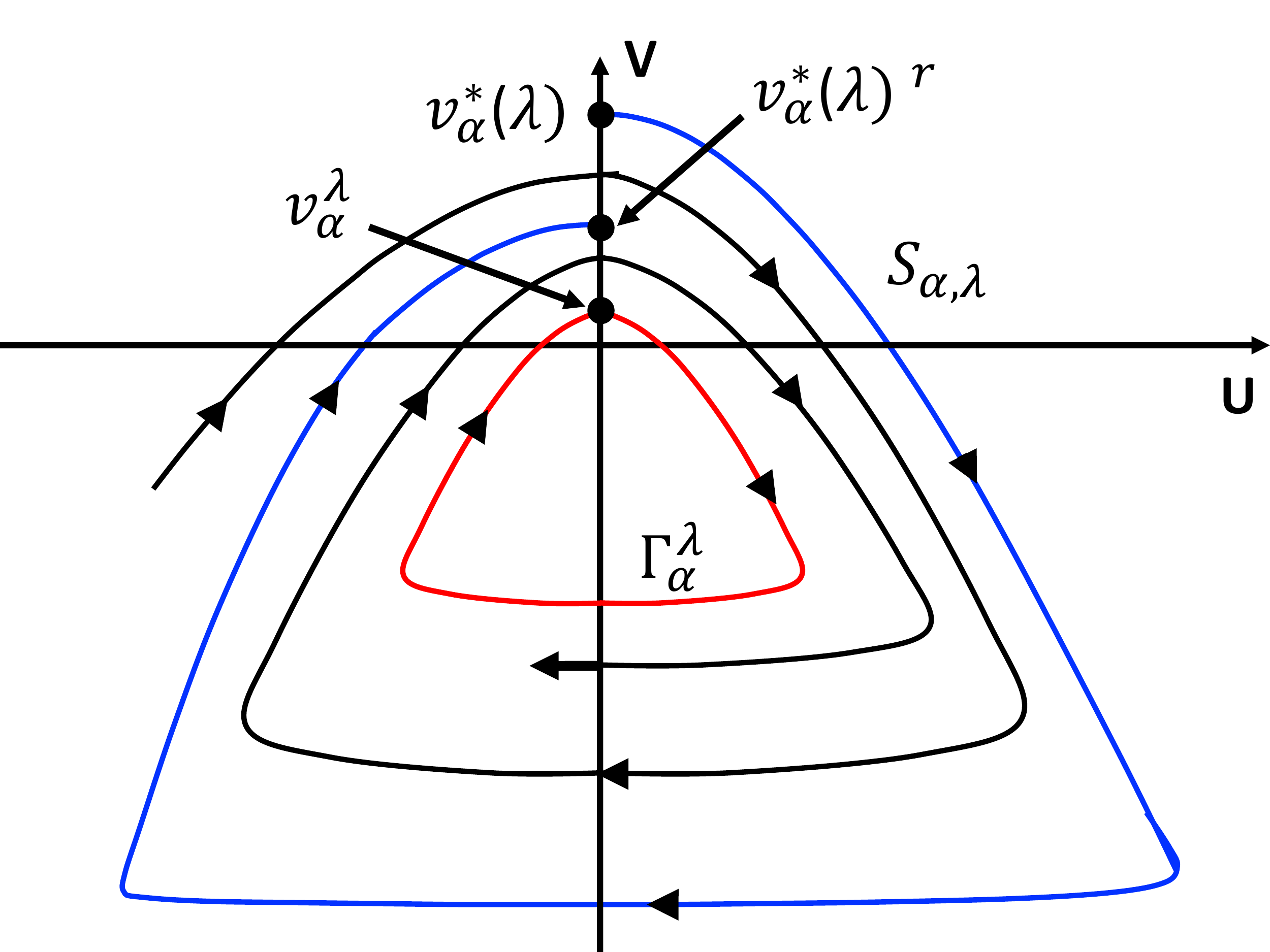}}\hspace{0.5cm}
  \subfloat[$\lambda> \lambda_c(\alpha)$]{\includegraphics[width=7.5cm]{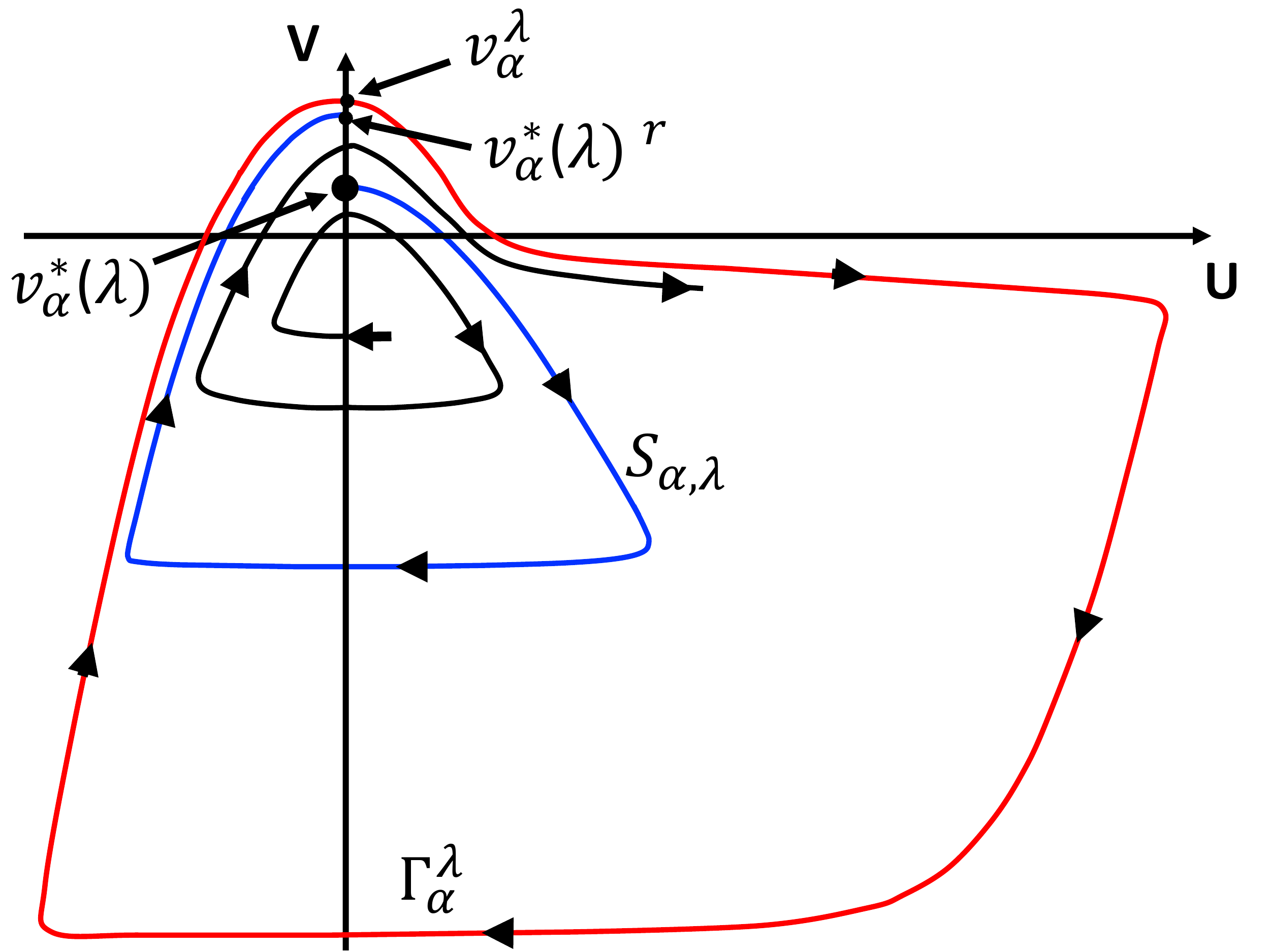}} 
  \caption{A qualitative representation of trajectories of  system (\ref{normal_par}) for $\delta>0$ as they approach the limit cycle ${\Gamma}^{\lambda}_{\alpha}$ (red), and  location of the separatrix $S_{\alpha,\lambda}$ (blue) with respect to ${\Gamma}^{\lambda}_{\alpha}$ upon variation of $\lambda$ for $\alpha>\alpha_c$. (A) For $\lambda< \lambda_c(\alpha)$, ${\Gamma}^{\lambda}_{\alpha}$ lies below  $S_{\alpha,\lambda}$. (B)  For $\lambda> \lambda_c(\alpha)$, ${\Gamma}^{\lambda}_{\alpha}$ lies above  $S_{\alpha,\lambda}$.}
  \label{separatrix_limitcycle}
\end{figure}
 The parameter space $(\alpha, \lambda)$ can be thus divided into three distinct regions, namely regime I $:=\{(\alpha, \lambda): 0<\alpha\leq \alpha_c ~ \textnormal{and} ~\lambda_H(\alpha)< \lambda \leq 0\}$,   regime II $:=\{(\alpha, \lambda): \alpha_c<\alpha \leq \alpha_s ~ \textnormal{and} ~\lambda_H(\alpha)< \lambda\leq  \lambda_c(\alpha)\}$  and regime III $:= \{(\alpha, \lambda): \alpha_c<\alpha \leq \alpha_s ~ \textnormal{and} ~ \lambda_c(\alpha)< \lambda \leq 0 \}$ (see  figure \ref{canard_explosion}).

\subsubsection{An analytical approximation of ${v_{\alpha}^*(\lambda)}$ and $\lambda_c(\alpha)$} 

In this subsection, we study the relative position of $S_{\alpha, \lambda}$ with respect to the level curve of $u^{iv}=0$ near the point $(0, v_{\alpha}^*(\lambda))$. The level curve for $u^{iv}=0$  can be computed analytically. We will also find an analytical approximation of  $\lambda_c(\alpha)$ for sufficiently small $\delta$, and therefore the analysis in this subsection has some practical implications.

\begin{proposition}\label{approx_critthresh}  Assume $F_{13}>0$, $F_{111}<0$ and $\delta>0$ is sufficiently small. Then for $\alpha> \alpha_c$, $\lambda_c(\alpha)$  can be written as  $\lambda_c(\alpha) = \tilde{\lambda}_{\delta}(\alpha) +O(\delta^2)$, where
\[ \tilde{\lambda}_{\delta}(\alpha)= \frac{1}{F_{13}}\Big(\frac{2-\sqrt{4+2\delta^2 F_{111}^2}}{\delta^2 F_{111}} -\alpha\Big).\]
  \end{proposition}
   \begin{proof}
   We will begin by parametrizing $S_{\alpha, \lambda}$ by $\tau$ on its interval of definition $[0, \tilde{\tau}(\lambda)]$ (see subsection 3.3.1 for the notations) and computing its higher order $\tau$-derivatives at $\tau \to 0^+$. 
Denoting the $u$ component of the separatrix by $u_*$ and its derivatives by $u_*^{(k)}$ for $k \geq 1$, we will interpret $u_*^{(k)}(0)$ and $u_*^{(k)}(\tilde{\tau}(\lambda))$ as the right-hand and the left-hand derivatives of $u_*^{(k)}(\tau)$  at $\tau =0$ and $\tau =\tilde{\tau}(\lambda)$ respectively. 
To this end, we have from (\ref{normal_par}) that $u_*'(0)= v_{\alpha}^*(\lambda)$, 
\bes \label {highorder} u_*''(0) =   \delta(\alpha+F_{13}\lambda)u_*'(0),\  u_*'''(0) = \delta(\alpha+F_{13}\lambda)u_*''(0)(u_*'(0)-1)u_*'(0)
\ees and 
\bes \nonumber  u_*^{(iv)} (0) &=& \delta(\alpha+F_{13}\lambda)u_*'''(0)+
    (3u_*'(0)-1)u_*''(0)+
     \delta F_{111}{u_*'}^3(0)\\
\nonumber      &=& \delta F_{111}{v_{\alpha}^*(\lambda)}^3 +4 \delta(\alpha+F_{13}\lambda){v_{\alpha}^*(\lambda)}^2+
      (\delta(\alpha+F_{13}\lambda)^3 -2 (\delta(\alpha+F_{13}\lambda)) v_{\alpha}^*(\lambda)\\
   \label{fourthlamb}    &=& \delta {v_{\alpha}^*(\lambda)} \Big[F_{111} {v_{\alpha}^*(\lambda)}^2+4  (\alpha+F_{13}\lambda){v_{\alpha}^*(\lambda)} +(\alpha+F_{13}\lambda) (\delta^2(\alpha+F_{13}\lambda)^2-2)\Big]
     \ees
     
  We recall that $|v_{\alpha}^*({\lambda}) - 1| = O(\delta)$, therefore we have from (\ref{highorder}) that $u_*''(0) = O({\delta})$ and $u_*'''(0) = O({\delta}^3)$ for all $\lambda_H(\alpha) < \lambda \leq 0$. 
Since $S_{\alpha, \lambda_c(\alpha)}$ forms a closed orbit, it can be viewed as a perturbation of   (\ref{periodic}) for some $k$ exponentially close to $0$ (see \cite{SK}). Hence we must have that for all $\lambda$ in a small neighborhood of $\lambda_c(\alpha)$, $u_*^{(iv)} (0) = O(\delta^m)$  for some $m \geq 4$.  This in turn would imply that for all such $\lambda$, $v_{\alpha}^*({\lambda})$ must be in $O(\delta^l)$ neighborhood of roots of $u_*^{(iv)} (0)=0$ for some $l>0$. Moreover, since $v_{\alpha}^*({\lambda})$ decreases with $\lambda$, there exists some $p\geq 1$ such that $|{v_{\alpha}^*({\lambda})} -1|=O(\delta^p)$ for $\lambda$ in a small neighborhood of $\lambda_c(\alpha)$. We will now consider the roots of $ u_*^{(iv)} (0)=0$  to obtain an estimate for $\lambda_c(\alpha)$.  
  Note that the equation   \bes \label{approxcanard} F_{111}v^2+4(\alpha+F_{13}\lambda)v+(\alpha+F_{13}\lambda)(\delta^2{(\alpha+F_{13}\lambda)}^2-2)=0
\ees
has real roots (positive) if and only if 
\bes
\nonumber 4(\alpha+F_{13}\lambda) +F_{111} (2-\delta^2{(\alpha+F_{13}\lambda)}^2) &\geq &0,\\
\nonumber \textnormal{i.e. }\ \alpha+F_{13}\lambda &\geq& \frac{2-\sqrt{4+2\delta^2 F_{111}^2}}{\delta^2 F_{111}},\\
\label{approxlamb}\textnormal{i.e. } \lambda &\geq & \frac{1}{F_{13}}\Big(\frac{2-\sqrt{4+2\delta^2 F_{111}^2}}{\delta^2 F_{111}} -\alpha\Big) := \tilde{\lambda}_{\delta}(\alpha).
\ees

Hence (\ref{approxcanard}) has two positive roots  $\hat{v}^{\pm}(\lambda)$ if $\tilde{\lambda}_{\delta}(\alpha)\leq \lambda \leq 0$, and no roots if $\lambda<\tilde{\lambda}_{\delta}(\alpha)$, where
\bess \hat{v}^{-}(\lambda) &=& \frac{2(\alpha +F_{13}\lambda)}{F_{111}}\Big[-1+ \sqrt{1-\frac{(\delta^2(\alpha+F_{13}\lambda)^2 -2)F_{111}}{4(\alpha+F_{13}\lambda)}}\Big] \ \textnormal{and}\\
 \hat{v}^{+}(\lambda) &=& \frac{2(\alpha +F_{13}\lambda)}{F_{111}}\Big[-1- \sqrt{1-\frac{(\delta^2(\alpha+F_{13}\lambda)^2 -2)F_{111}}{4(\alpha+F_{13}\lambda)}}\Big].
\eess

At $\lambda=\tilde{\lambda}_{\delta}(\alpha)$, we have
\bess \nonumber \hat{v}^-(\tilde{\lambda}_{\delta}(\alpha)) &=&  \hat{v}^+(\tilde{\lambda}_{\delta}(\alpha)) =\frac{4}{2+\sqrt{4+2\delta^2 F_{111}^2}} \\
 & = & 1 -\frac{\delta^2}{2}  F_{111}^2 +O(\delta^4) \to 1 ~ \textnormal{as} ~ \delta \to 0. 
\eess
 The roots $\hat{v}^{\pm}(\lambda)$ have the property that $\hat{v}^{-}(\lambda) < \hat{v}^-(\tilde{\lambda}_{\delta}(\alpha))< \hat{v}^{+}(\lambda)$ and that $\hat{v}^{-}(\lambda)$ decreases, while $\hat{v}^{+}(\lambda)$ increases monotonically on $(\tilde{\lambda}_{\delta}(\alpha), 0]$.  Furthermore, it can be shown that  
 \bes \label{asymp} \hat{v}^{-}(\lambda) =1-O(\delta) \ \textnormal{and} \  \hat{v}^{+}(\lambda) =1+O(\delta) \  \textnormal{if and only if} \ \lambda =\tilde{\lambda}_{\delta}(\alpha)  +O(\delta^2).
 \ees
  Combining (\ref{asymp}) with the fact that  $|{v_{\alpha}^*({\lambda})} -1|=O(\delta^l)$ for some $l\geq 1$ if  $\lambda$ is in a small neighborhood of $\lambda_c(\alpha)$ yields that ${v_{\alpha}^*({\lambda}_c(\alpha))}$ will be within $O(\delta^p)$ neighborhood of $\hat{v}^{\pm}(\lambda)$ for some $p>0$ only if $\lambda_c (\alpha) = \tilde{\lambda}_{\delta}(\alpha) +O(\delta^2)$. This completes the proof.
 
   \end{proof}
   
      \begin{figure}[h!]     
  \centering 
  \subfloat[$\lambda=-0.4 > \tilde{\lambda}_{\delta}(\alpha)$]
  {\includegraphics[width=7.5cm]{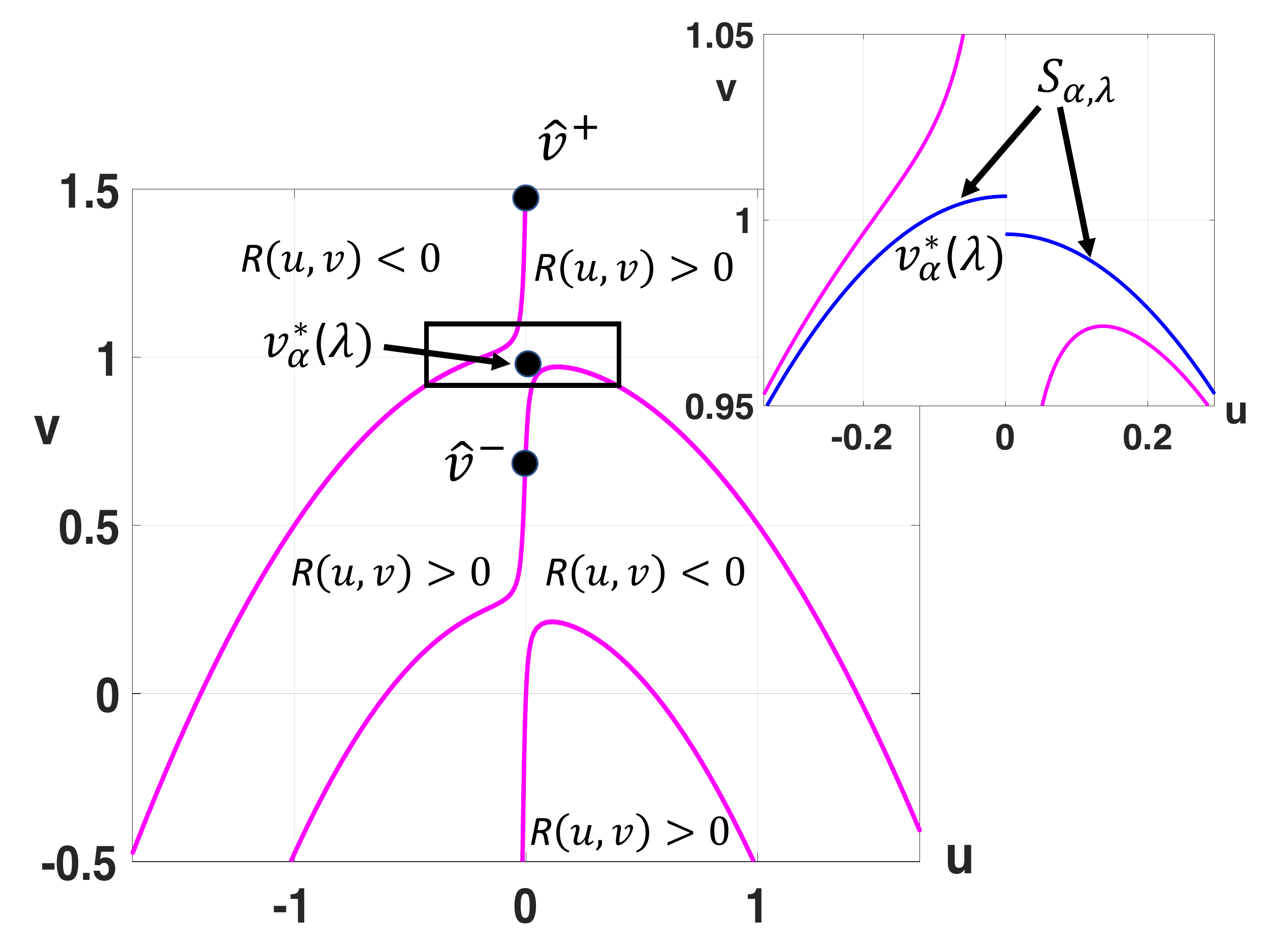}}\hspace{0.5cm}
  \subfloat[$\lambda=-0.9 < \tilde{\lambda}_{\delta}(\alpha)$]{\includegraphics[width=7.5cm]{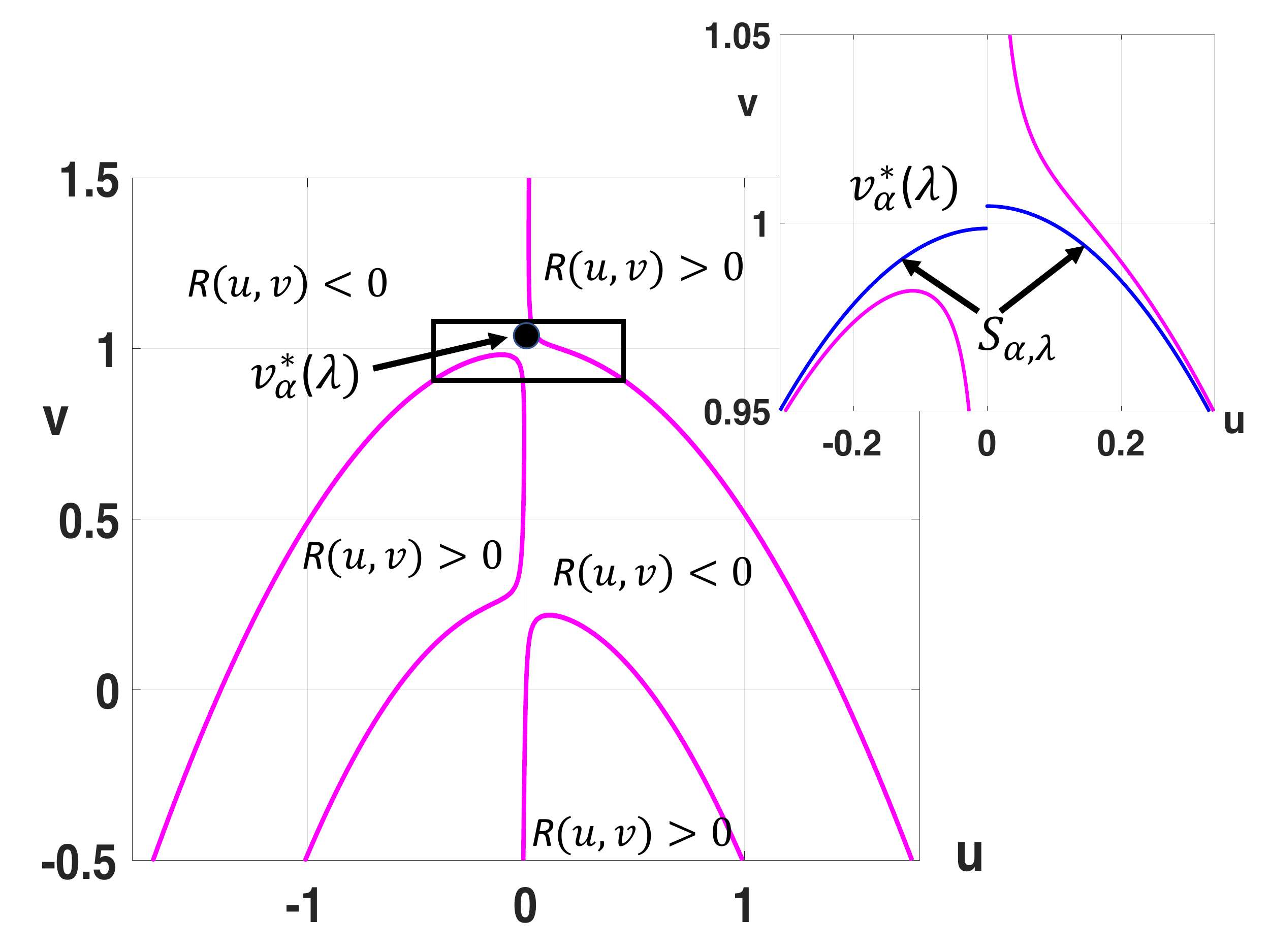}} 
    \caption{The level curve $R(u,v)=0$ (magenta) for different values of $\lambda$ with other parameter values being $\delta=0.078,$ $\alpha=0.4613$, $F_{13} = 0.16454$ and $F_{111} = -0.6833$. Here $\tilde{\lambda}_{\delta}(\alpha)=-0.728$. The insets show the position of the separatrix $S_{\alpha,\lambda}$ of  system (\ref{normal_par}) with respect to  $R(u,v)=0$ locally near  $(0, {v_{\alpha}^*(\lambda)})$.  (A) The separatrix  lies in the region $R(u, v)>0$ for $u=o(1)$. (B) The separatrix lies in the region $R(u, v)<0$ for $u=o(1)$.}
  \label{separatrix_fourth}
\end{figure} 

We remark that the proof of Proposition \ref{approx_critthresh} also yields explicit bounds on $v_{\alpha}^*(\lambda)$. Indeed,  since ${v_{\alpha}^*(\lambda)}$ decreases with $\lambda$ and $|{v_{\alpha}^*(\lambda)} -1|=O(\delta)$,  it follows from the proof of Proposition \ref{approx_critthresh} that ${v_{\alpha}^*(\lambda)} \in (\hat{v}^-(\lambda), \hat{v}^+(\lambda))$ on $[\lambda_{\delta}(\alpha), 0]$ for some $\lambda_{\delta}(\alpha) \geq  \tilde{\lambda}_{\delta}(\alpha)$ and ${v_{\alpha}^*(\lambda)} \geq {\hat{v}^{\pm} (\tilde{\lambda}_{\delta}(\alpha))}$ on $(\lambda_H(\alpha), \tilde{\lambda}_{\delta}(\alpha))$ for $\alpha>\alpha_c$. 
Moreover, we can compare the position of $S_{\alpha,\lambda}$ with respect to the level curve of $u^{(iv)}=0$ locally near the point $(0, {v_{\alpha}^*(\lambda)})$. 
Denoting the level curve of $u^{(iv)}=0$ by $R(u,v)=0$, it turns out that ${v_{\alpha}^*(\lambda)}$  lies in the region $\{(u,v): R(u,v)<0 , u=o(1) , v =1+o(1)\}$ if $\lambda < \tilde{\lambda}_{\delta}(\alpha)$ and in the region $\{(u, v): R(u,v)>0, u=o(1),  v =1+o(1)\}$ if $\tilde{\lambda}_{\delta}(\alpha)<\lambda<0$  for sufficiently small $\delta>0$ as shown in figure \ref{separatrix_fourth}.  The curve $\tilde{\lambda}_{\delta}(\alpha)$  is shown in figure \ref{canard_explosion}.

    \begin{figure}[h!]     
  \centering 
  {\includegraphics[width=9.5cm]{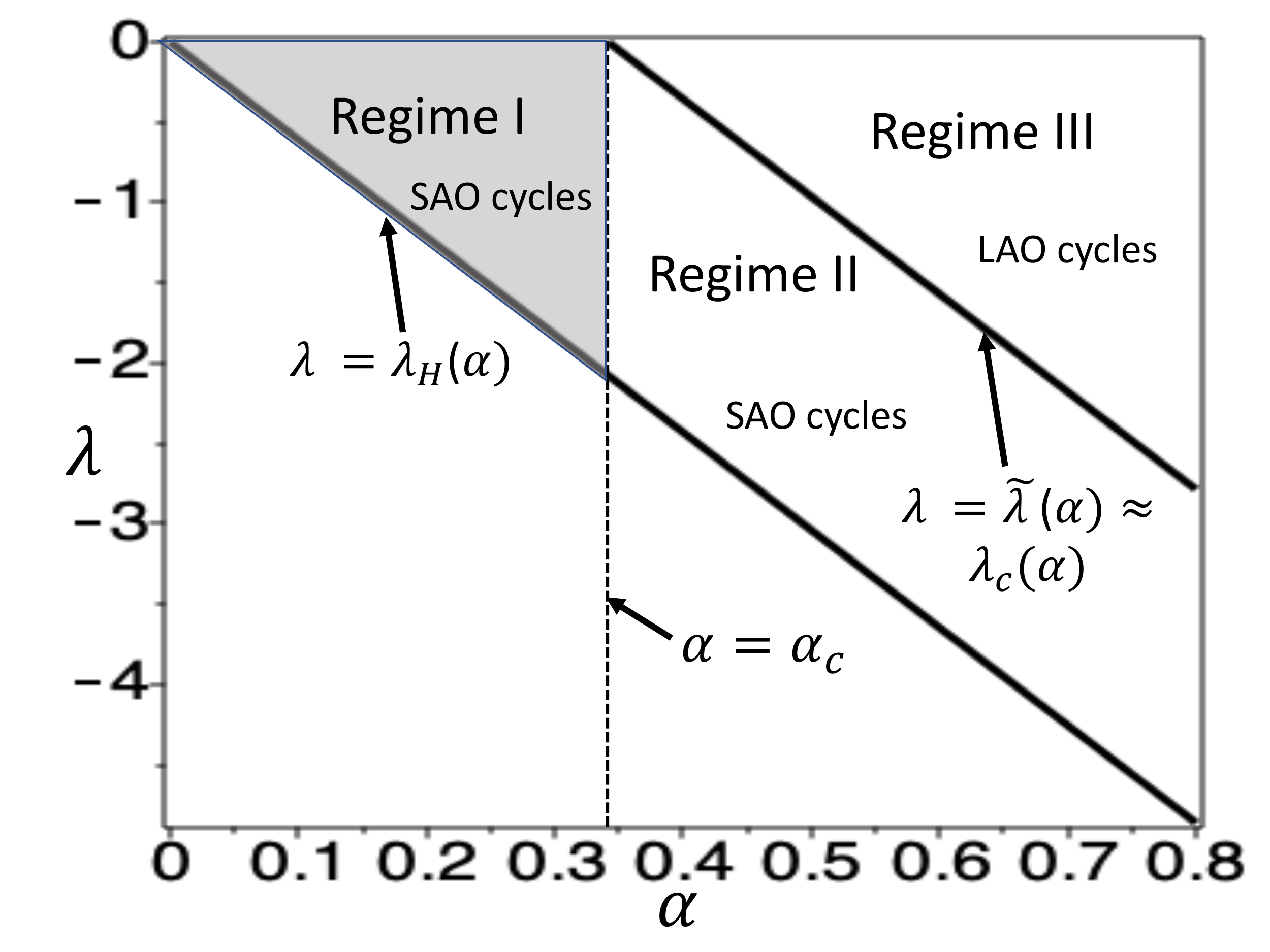}}
  \caption{A two-parameter bifurcation diagram of system (\ref{normal_par})  in $\alpha$ and $\lambda$. SAOs occur in Regimes I and II, LAOs occur in Regime III. $\lambda_H(\alpha)$: Hopf curve, $\lambda_c(\alpha)$: curve along which canard explosion occurs. The other parameter values are $\delta=0.078,$ $F_{13} = 0.16454$ and $F_{111} = -0.6833$.}
  \label{canard_explosion}
\end{figure}

\subsubsection{Existence of a separatrix in the normal form (\ref{normal2})} Extending the analysis of system (\ref{normal_par}), we will now find a set of necessary and sufficient conditions for a Type II oscillation (large amplitude oscillation) in system (\ref{normal2}). 
We first note that for $\alpha>0$, any trajectory of (\ref{normal2}) that starts above the $uv$-plane,  eventually goes below it, since $w'(\tau)<0$ if $w(\tau)>0$ (follows from the $w$-equation in (\ref{normal2})). Hence we will start with an initial value such that $w(0)<0$. We also note that if $w(\tau)<0$, then as long as $u^2(\tau)<2H_3/H_{11} |w(\tau)|$, a trajectory in a neighborhood of $p_e=(0, 0, 0)$ spirals up  along the $w$-axis while approaching towards it. Since $p_e$ has a two-dimensional unstable manifold, the trajectory then spirals out along $W^u(p_e)$ with increasing values of $|u|$ and $|v|$. This then  leads to larger and larger negative average values of $w'(\tau)$, causing the trajectory to descend.

Consider the surface $\mathcal{G}(u, \lambda)=\{(u, v_-^{\lambda}(u), \lambda): u_1^{\lambda}\leq u\leq u_2^{\lambda},\ \lambda_H(\alpha)\leq \lambda<0\}$, where $v_-^{\lambda}(u)$ is defined by (\ref{roots}) and $ u_{1,2}^{\lambda}$ are roots of $f_{\lambda,\delta}'(u) =2$.  Let $\Omega(\alpha)$ be the region defined by  $\Omega(\alpha)=\{(a, b, c)\in \mathbb{R}^3: b\leq 1-a^2/2, \ c \in [\lambda_H(\alpha), 0]\}$. We will say that a trajectory $\gamma$ of system (\ref{normal2}) with initial data \bes \label{intdata}  
 0<|u(0)|, |v(0)| \ll1, \ \lambda_c(\alpha)<w(0)<0
 \ees
 exhibits  a Type II oscillation if it exits $\Omega(\alpha)$  and crosses $\mathcal{G}$ with $ {\partial \gamma}/{\partial \bf{n}}>0$, where ${\bf{n}}$ is the outward normal vector to ${\mathcal{G}}$.  
For $0<\alpha\leq \alpha_c$,  it can be shown that the closed cylindrical region $\Gamma^0_{\alpha}\times [\lambda_H(\alpha),0]$ forms a positively invariant set for system (\ref{normal2}), where  $\Gamma^0_{\alpha}$ is the limit cycle of  system (\ref{normal_par}) with $\lambda=0$. Since $p_e$ is unstable, by Theorem \ref{hopfthm} we know that system (\ref{normal2}) asymptotically approaches the stable limit cycle $\Gamma_{\alpha}$. Hence we must have $\Gamma_{\alpha} \subset \Gamma^0_{\alpha}\times [\lambda_H(\alpha),0]$. Therefore for $\alpha$ in this range, a trajectory that starts at $(u(0), v(0), w(0))$ such that $0<|u(0)|, |v(0)| \ll1, \lambda_H(\alpha)<w(0)<0$ stays bounded within $\Omega(\alpha)$ and never exhibits a Type II oscillation. Henceforth, we assume $\alpha_c<\alpha \leq \alpha_s$, where $\alpha_s>0$ is such that $\Gamma_{\alpha}$ persists as the unique asymptotic attractor. We state the main result below.

\begin{theorem}\label{thmdistint} Assume that $F_{13}>0$, $F_{111}<0$, $H_3<0$ and $H_{11}<0$. Then for  every  $\alpha \in (\alpha_c, \alpha_s]$ and $\delta>0$ sufficiently small, a solution $(u(\tau), v(\tau), w(\tau))$ of system (\ref{normal2}) with $0<|u(0)|,\  |v(0)| \ll 1$ and $w(0)> \lambda_c(\alpha)$ exhibits a Type II oscillation if and only if  there exists some $\tau^+>0$ such that  $ v(\tau^+)>v_{\alpha}^*(w(\tau^+))$ with  $w(\tau^+) >\lambda_c(\alpha)$ on the plane $\{u=0\}$, where $v_{\alpha}^*(w(\tau^+))$ is defined by (\ref{threshold}).
\end{theorem}

 \begin{proof} Let $\tilde{\tau}(\alpha)>0$ be the first time a trajectory  with initial data (\ref{intdata}) intersects with $\mathcal{G}$.  We will let $\tilde{\tau}(\alpha) = \infty$ in situations when the trajectory never  intersects with $\mathcal{G}$. Without loss of generality, we assume that $w(0)$ is the absolute maximum of $w(\tau)$ on $[\hat{\tau}, \tilde{\tau}]$, where $\hat{\tau} =  \sup\{ \tau<0: w(\tau)<\lambda_H(\alpha) \}$. As the trajectory spirals out along the two-dimensional unstable manifold $W^u(p_e)$,  
it follows from (\ref{flow1}) that the time taken by the fast variables $(u, v)$  to make one complete helical turn is  approximately $2\pi/\nu$, and that during this period, $w(\tau)$ descends by approximately
 \[ \frac{-2\delta w_0 H_3 \pi}{\nu} -\frac{\delta^2}{\nu} H_{11} (H_3+\alpha) \pi \Big(A+\frac{4\vartheta^2 u_0^2+(2v_0 +\alpha \delta u_0 )^2}{8 \vartheta^2 \delta (\alpha-H_3)}\Big)= O(\delta),\] where $(u_0, v_0, w_0)$ is the position of the trajectory at the beginning of the turn and $A$, $\nu$ are defined by (\ref{flow1}).   Hence we will treat $w(\tau)$ to be constant during every helical turn (one possibility is  setting $w(\tau)$ to be its average value during that period). Hence to the leading order, we have $w(\tau) = \lambda +O(\delta)$ for some $\lambda <0$ on every such interval. 
Therefore the dynamics of the fast variables of system (\ref{normal2}) can be approximated by system (\ref{normal_par}) over every such period. 

As the trajectory descends, the amplitude of oscillations of the fast variables initially increases. 
We will show that $w(\tau)\geq \lambda_H(\alpha)$ for $\tau \in [0,\tilde{\tau})$. To see this, 
let $r^i_1<s^i_1<r^i_2<s^i_2$ be the locations of relative extrema of $w(\tau)$ during the $i$th helical turn, where $i \in \mathbb{N}$, $w(r^i_{1,2})$ and $w(s^i_{1,2})$ correspond to relative minima and relative maxima respectively. As long as the trajectory descends while exhibiting Type I oscillations, we must have $w(r^i_{2})<w(r^i_{1})<w(r^{i-1}_{2})<w(r^{i-1}_{1})$. 
We note from system (\ref{normal2}) that $w'>0$ if $|u|=o(1)$ and $|w| \gg 1$. Hence, $w$ must attain its minimum before the oscillations in $u$ decrease to $o(1)$. Suppose that absolute minimum of $w(\tau)$ occurs at the $n$th helical turn and  that $w(r^n_{1})$ is the minimum.  If possible, let  $w(r^n_{1})<\lambda_H(\alpha)$. Then from (\ref{normal2}), we obtain that 
\[ u^2(r^n_{1})=  -\frac{2 H_3}{H_{11}}w(r^n_{1})>-\frac{2 \lambda_H H_3 }{H_{11}},\ \textnormal{for} \ i=1, 2, \ldots n-1.\]
This implies that the amplitude of oscillation of $u(\tau)$ during the $n$th turn must be greater than $\sqrt{-{2 H_3}\lambda_H/{H_{11}}}=O(1)$. However, as $w(\tau)$ decreases and approaches $\lambda_H(\alpha)$, it follows from (\ref{normal_par}) that the amplitude of oscillations of $u(\tau)$ and $v(\tau)$ must be $o(1)$, thus contradicting the above assumption.   
Hence we must have $w(r^n_{1})\geq \lambda_H(\alpha)$  and therefore $w(\tau)\geq \lambda_H(\alpha)$ for $\tau \in [0,\tilde{\tau})$.

Let $\alpha \in (\alpha_c, \alpha_s]$ and consider a solution of  (\ref{normal2})  with initial data (\ref{intdata}).  To determine whether the trajectory escapes the domain bounded by $\Omega$ and crosses $\mathcal{G}$, consider the first return map on the Poincar\'e section $\{u=0\}$ such that $u'(\tau)>0$. Let $E_{\alpha}^+ =\{{\tau}_1\in (0,\tilde{\tau}(\alpha)) : u({\tau}_1)=0, v({\tau}_1)>0  ~\textnormal{and}~ w({\tau}_1)\in [\lambda_H(\alpha), 0]\}$. 
 By the definition of $v_{\alpha}^*(\lambda)$ in (\ref{threshold}), for each   ${\tau}_1\in E_{\alpha}^+$, there exists $v_{\alpha}^*(w({\tau}_1))>0$  such that the fast variables governed by system (\ref{normal_par}) with  $\lambda=w({\tau}_1)$ will exhibit a Type II oscillation  if and only if $v({\tau}_1)>v_{\alpha}^*(w(({\tau}_1))$. 
 This in turn implies that the solution of (\ref{normal2}) must also make a Type II oscillation if there exists some ${\tau}^+\in E_{\alpha}^+$ such that $v({\tau}^+)>v_{\alpha}^*(w({\tau}^+))$. In such a scenario, the first return map is defined for $\tau\in (0, {\tau}^+]$ and is  monotonically increasing. On the other hand, if there does not exist  any ${\tau}_1\in  E_{\alpha}^+$ such that $v({\tau}_1)>v_{\alpha}^*(w({\tau}_1))$, then the first return map is defined for all $\tau \in (0, \infty)$. The trajectory never intersects with $\mathcal{G}$ and approaches $\Gamma_{\alpha}$  as $\tau \to \infty$, where $\Gamma_{\alpha} \subset \Gamma^{\lambda_c}_{\alpha}\times [\lambda_H, \lambda_c]$. 

Taking the relative position of limit cycles of  system (\ref{normal2}) with respect to $S_{\alpha, \lambda}$  into account (see subsection 3.3.2),  we also note that if $w(\tilde{{\tau}}_1)<\lambda_c(\alpha)$ for some $\tilde{{\tau}}_1\in E_{\alpha}^+$, then the fast variables of system (\ref{normal2}) will be bounded by the Type I limit cycle $\Gamma_{\alpha}^{w(\tilde{{\tau}}_1)}$. Consequently, $v(\tilde{{\tau}}_1)<v_{\alpha}^*(\lambda_c(\alpha))$ (see figure \ref{separatrix_limitcycle}). Furthermore by Lemma  \ref{lm1}, we know that $v_{\alpha}^*(\lambda)$ is monotonically decreasing with $\lambda$,  hence if $\tilde{{\tau}}_1$ exists such that  $w(\tilde{{\tau}}_1)<\lambda_c(\alpha)$, then it follows that $v({\tau})< v_{\alpha}^*({\lambda}_c)$ for all ${\tau}\in E_{\alpha}^+$. On the other hand, if $v({\tau}^+)>v_{\alpha}^*(w({\tau}^+))$ for some ${\tau}^+ \in E_{\alpha}^+$, then it is also clear that  $w({\tau}^+)>\lambda_c(\alpha)$.

\end{proof}

 Theorem \ref{thmdistint}  implies that system (\ref{normal2}) possesses a surface delimiting small amplitude oscillations and large excursions. This surface  will be defined by  $S_{\alpha,\lambda}\times \{\lambda \}$ and will be referred to as the separatrix of (\ref{normal2}). 
As long as a trajectory stays in the ``inner side" of this surface, it exhibits SAOs.  A large amplitude oscillation in (\ref{normal2}) will be initiated if the trajectory goes to the outer side of this surface. To generate $S_{\alpha,\lambda}\times \{\lambda\}$, we considered orbits of  (\ref{normal_par}) (integrated backward and forward in time) through $(u_0^*(\lambda), v_0^*(\lambda))$ until they intersected the positive $v$-axis for $\lambda \in (\lambda_H(\alpha), 0]$. The  generated surface  is shown in figure \ref{qsc}. 

   \begin{figure}[h!]     
  \centering 
\subfloat[]{\includegraphics[width=7.0cm]{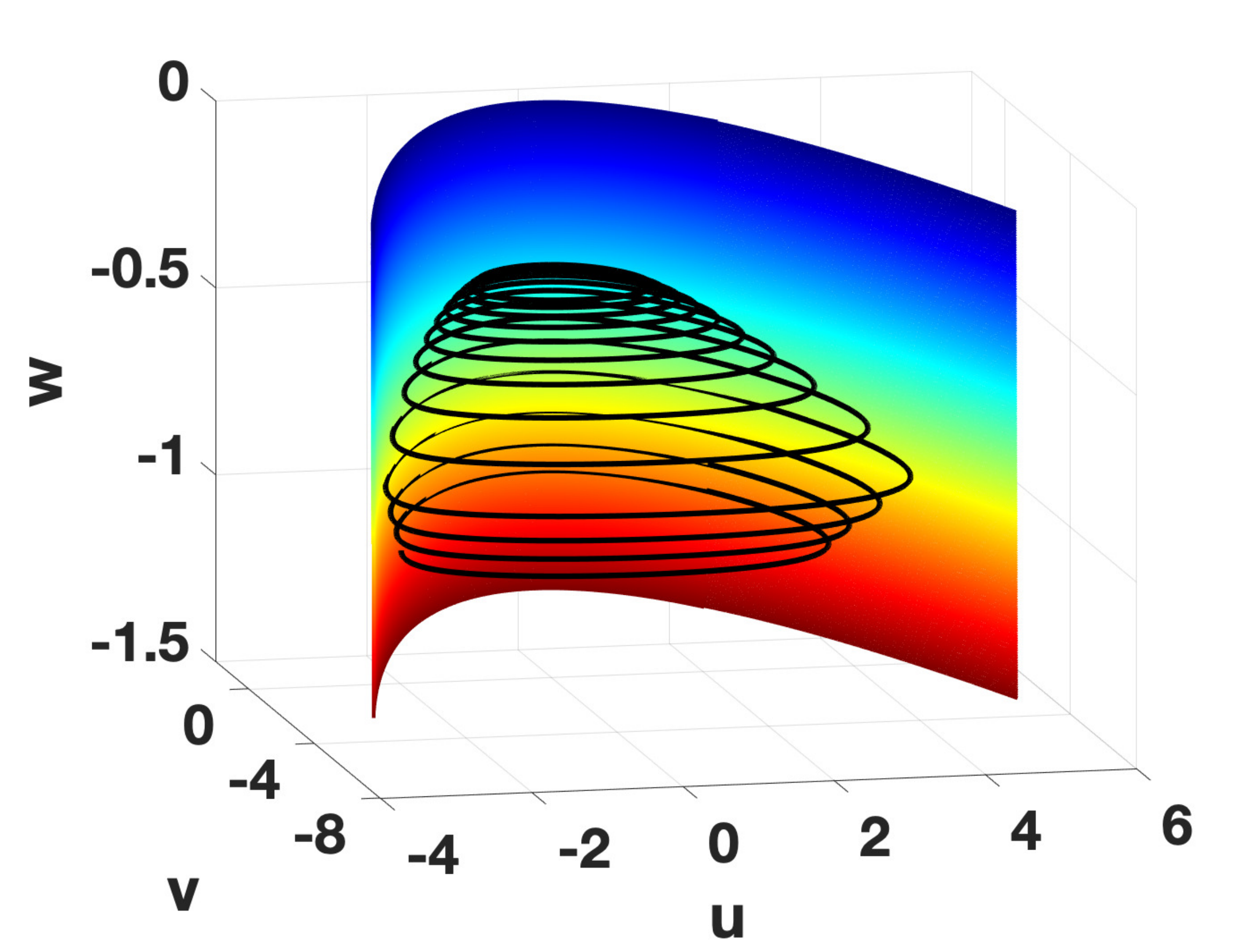}} \quad
   \subfloat[]
  {\includegraphics[width=7.0cm]{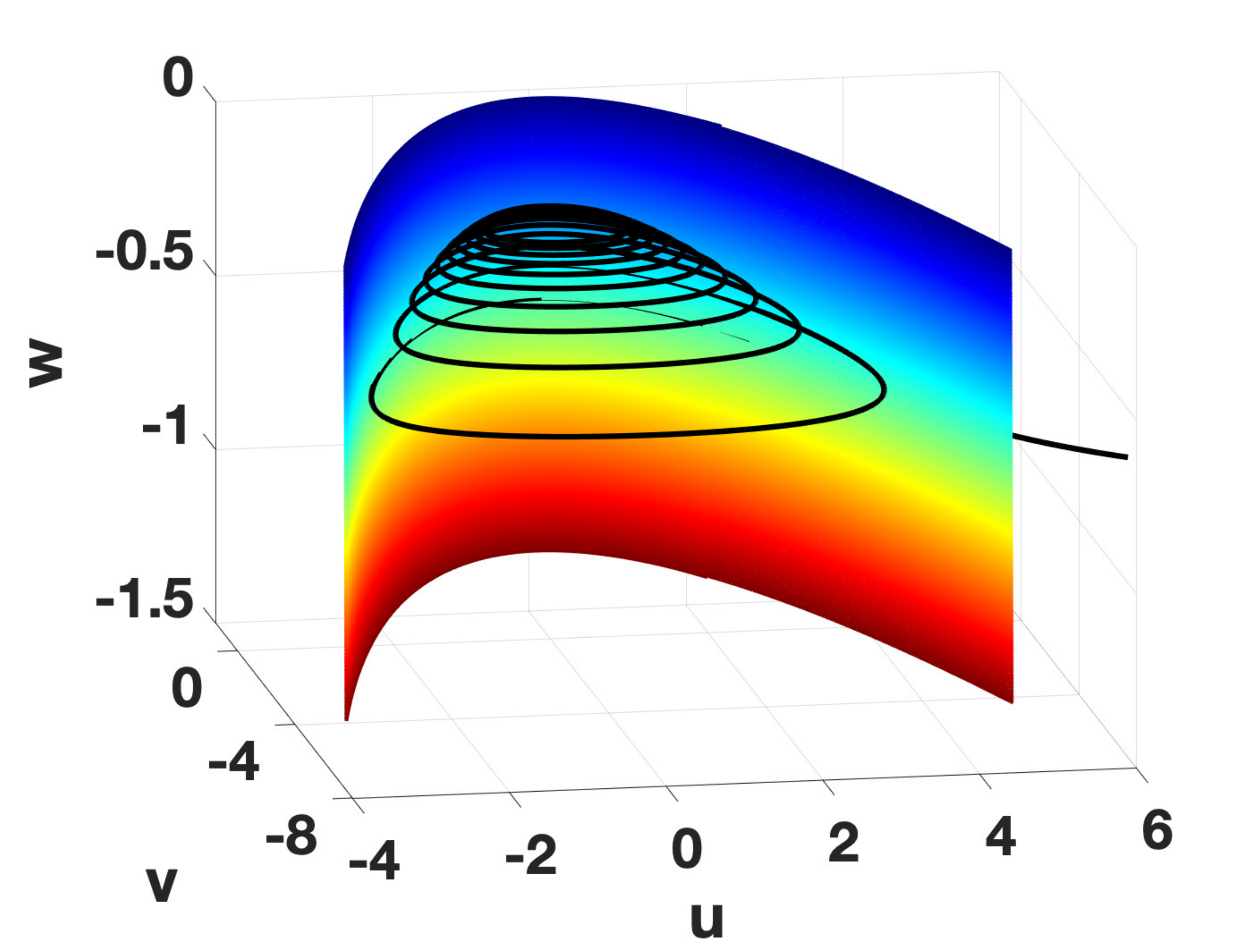}}\quad 
  \subfloat[]{\includegraphics[width=7.0cm]{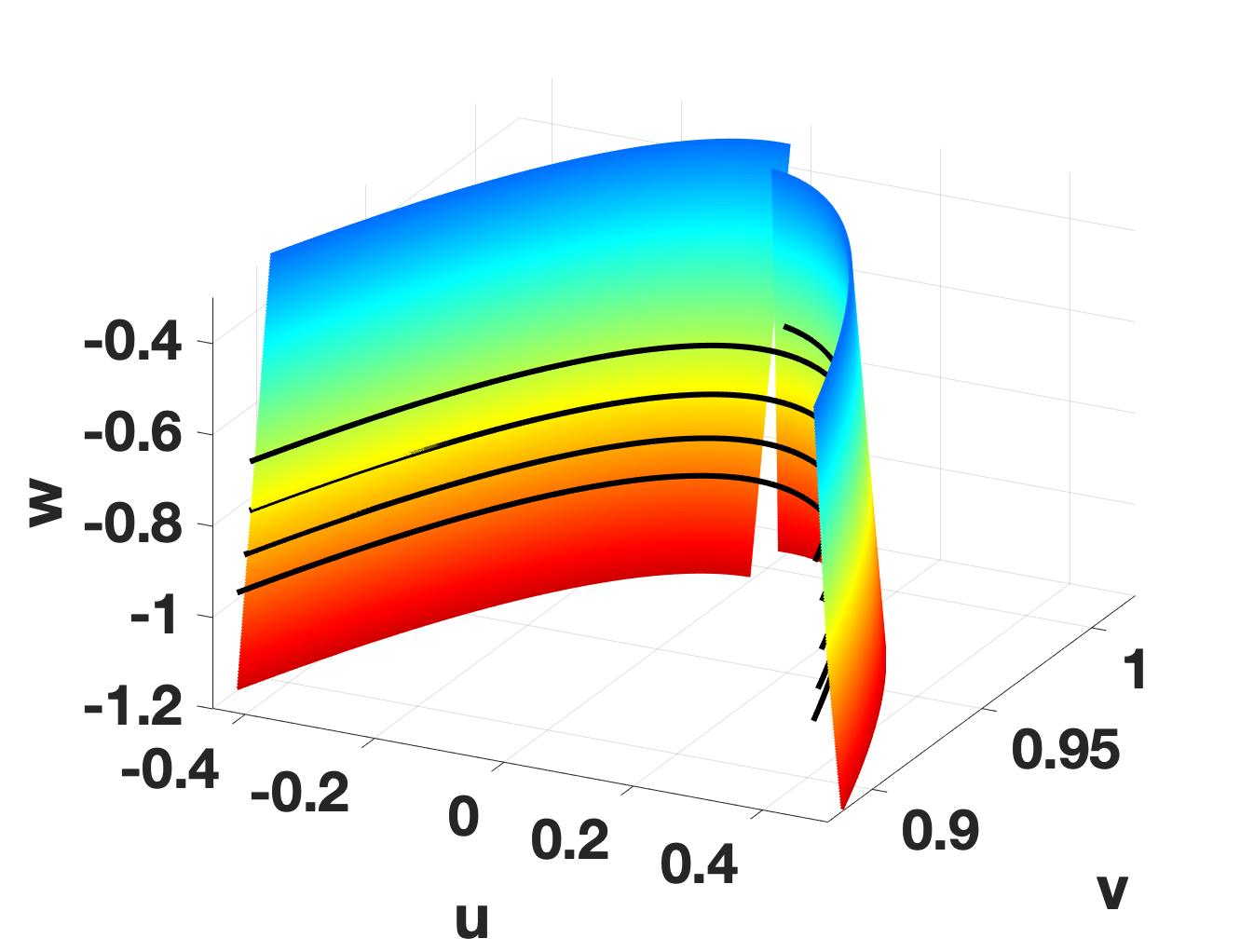}} \quad
   \subfloat[]
  {\includegraphics[width=7.0cm]{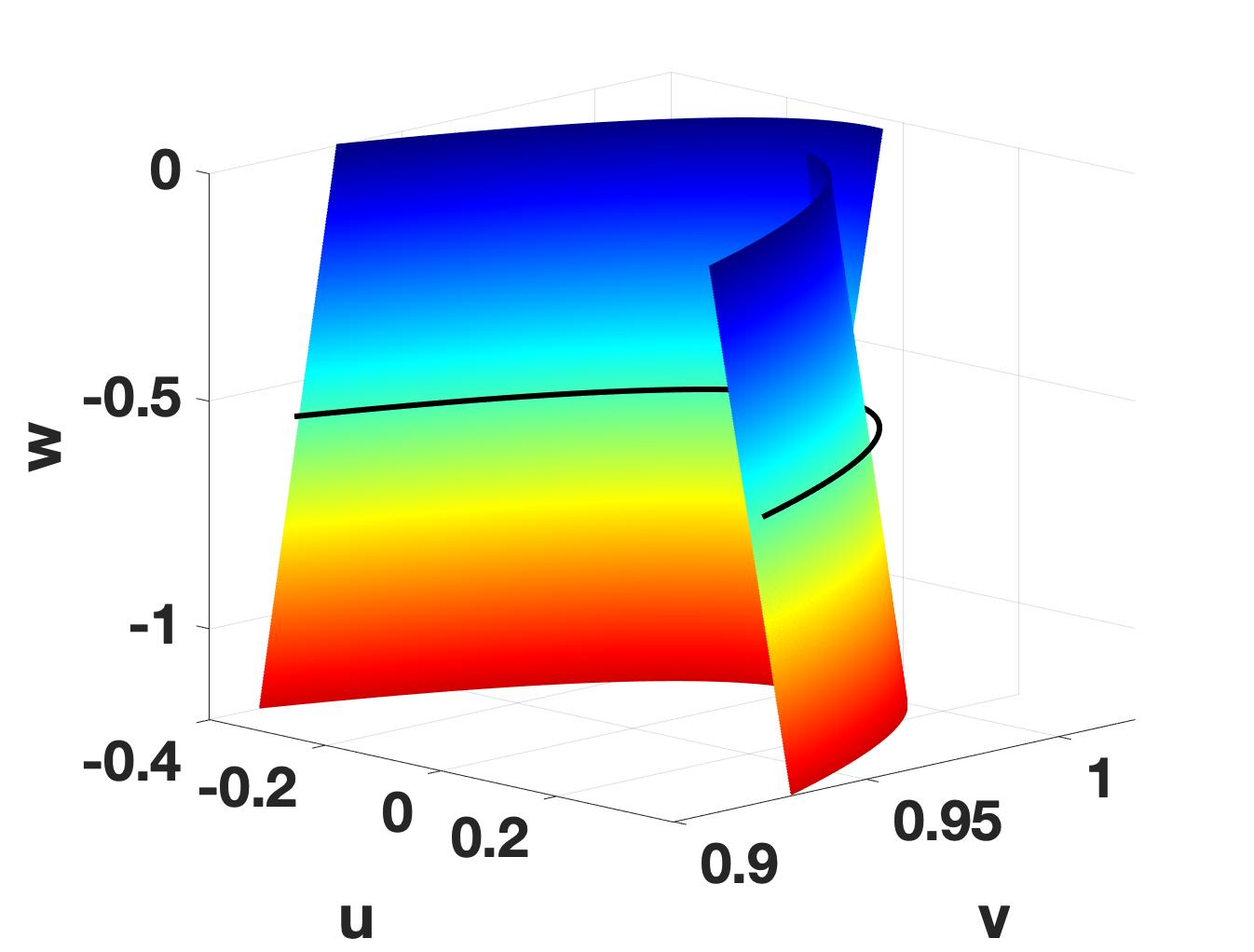}}\quad 
  \caption{(A)-(B) A view of the surface $S_{\alpha,\lambda}\times \{\lambda \}$  and an illustration of  dynamics of system (\ref{normal2}) governed by the location of a trajectory with respect to the surface. (C)-(D) Zoomed views of the surface and dynamics near $\{u=0\}$ plane. A trajectory exhibits SAOs (type I) if it lies in the ``inner side" of the surface. A large amplitude oscillation (type II) is initiated if the trajectory escapes to the ``outer side" of the surface. Here $\alpha=0.4613$, $\zeta=0.001$ and the other parameter values are as in  (\ref{parvalues01}). 
  }
  \label{qsc}
\end{figure}

\begin{remark} The characterization of the dynamics in a neighborhood of the origin and the existence of a separatrix for system (\ref{normal2}) via Theorem \ref{thmdistint} is also valid for system (\ref{normal1}) in a neighborhood containing the equilibrium and the folded node singularity for sufficiently small $\zeta>0$. To extend the separatrix globally for (\ref{normal1}), one may have to consider different charts and connect the flow across the charts, similar to the approaches used in  \cite{KS, KS1}.
\end{remark}

\section{Numerical analysis of the normal form}
As an illustration, we consider system (\ref{nondim3}) near the singular Hopf bifurcation with parameter values as in Section 2, namely
\bes \label{parvalues}
\beta_1=0.25, \ \beta_2 =0.35,\ c=0.4,\ d=0.21,\ \alpha_{12} =0.5, \ \alpha_{21}=0.1.
\ees 
As earlier, the intraspecific competition coefficient $h$ (and hence $\alpha$) is treated as our varying parameter. In the singular limit of system (\ref{nondim3}), the singular Hopf point $(\bar{x}, \bar{y}, \bar{z})$ has coordinates  $\approx (0.3381, 0.0903, 0.3497)$ and FSN II bifurcation occurs at $\bar{h}\approx 0.7785$. The coefficients of the normal form (\ref{normal2}), calculated using (\ref{del}), are given below:
\bes \label{parvalues01}
\delta \approx 2.4649 \sqrt{\zeta},\ F_{13} \approx 0.16454,\ F_{111} \approx -0.6833,\ H_3\approx-0.0145,\\
\nonumber  H_{11}\approx -0.065068.
\ees
and
\bess
\alpha(h)= \frac{1.5996(h-\bar{h})}{\zeta}-0.25779.
\eess

To study the dynamics near the singular Hopf bifurcation, $\zeta$ must be chosen sufficiently small. In system (\ref{nondim3}), the analysis was performed for $\zeta =0.01$. Here we choose $\zeta =0.001$ to obtain a better approximation to the singular limit. However, all the findings that we obtain for smaller values of $\zeta$ also hold for $\zeta=0.01$.  
Choosing $\zeta =0.001$ yields $\delta \approx 0.078$ and 
$
\alpha(h) = 1599.63h-1245.62.
$
  \begin{figure}[h!]     
  \centering 
  \subfloat[Transient MMO dynamics as it approaches $\Gamma_{\alpha}$.]{\includegraphics[width=7.5cm]{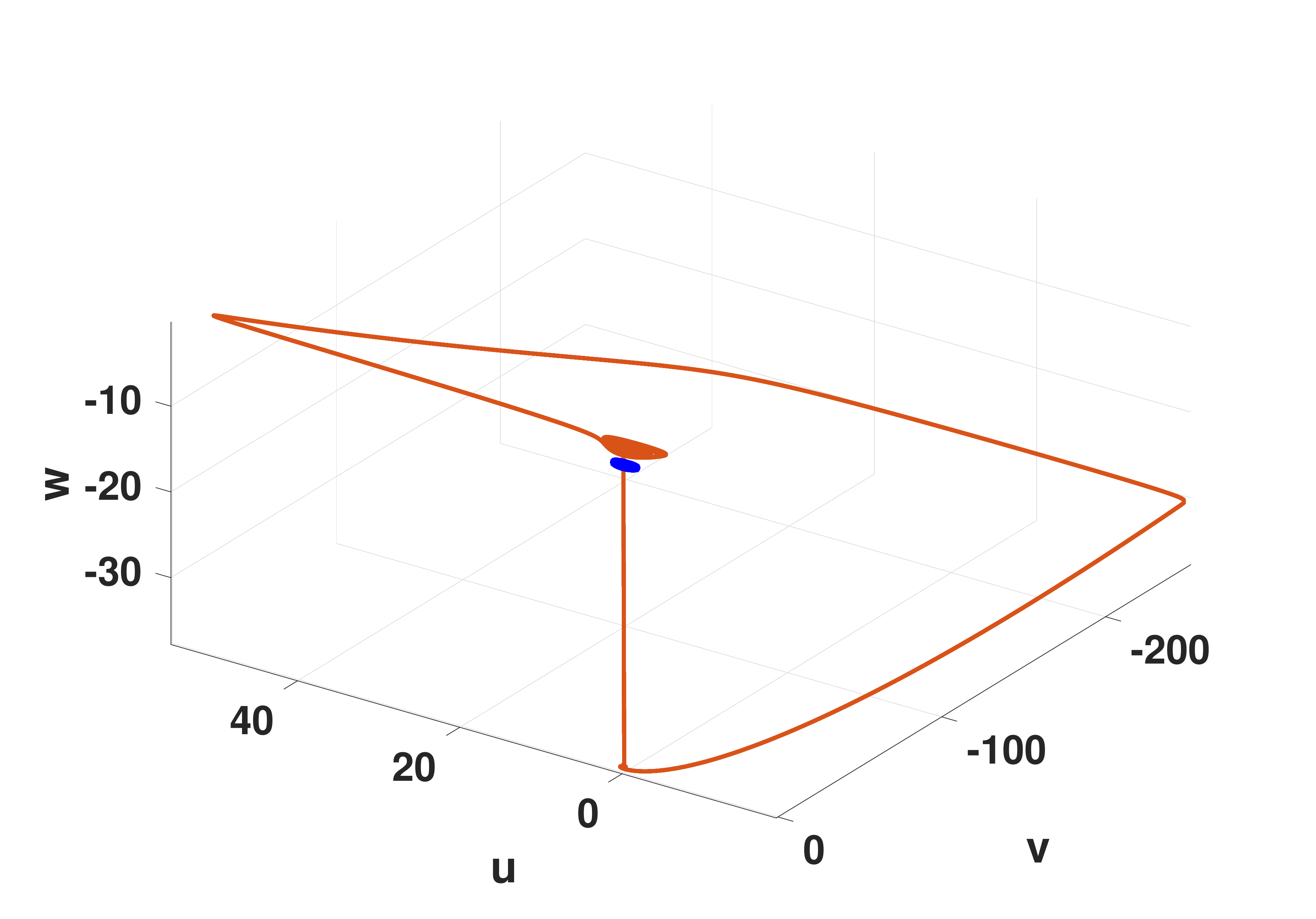}}\qquad
    \subfloat[A zoomed view near  the  origin projected on the $uv$-plane.]{\includegraphics[width=7.5cm]{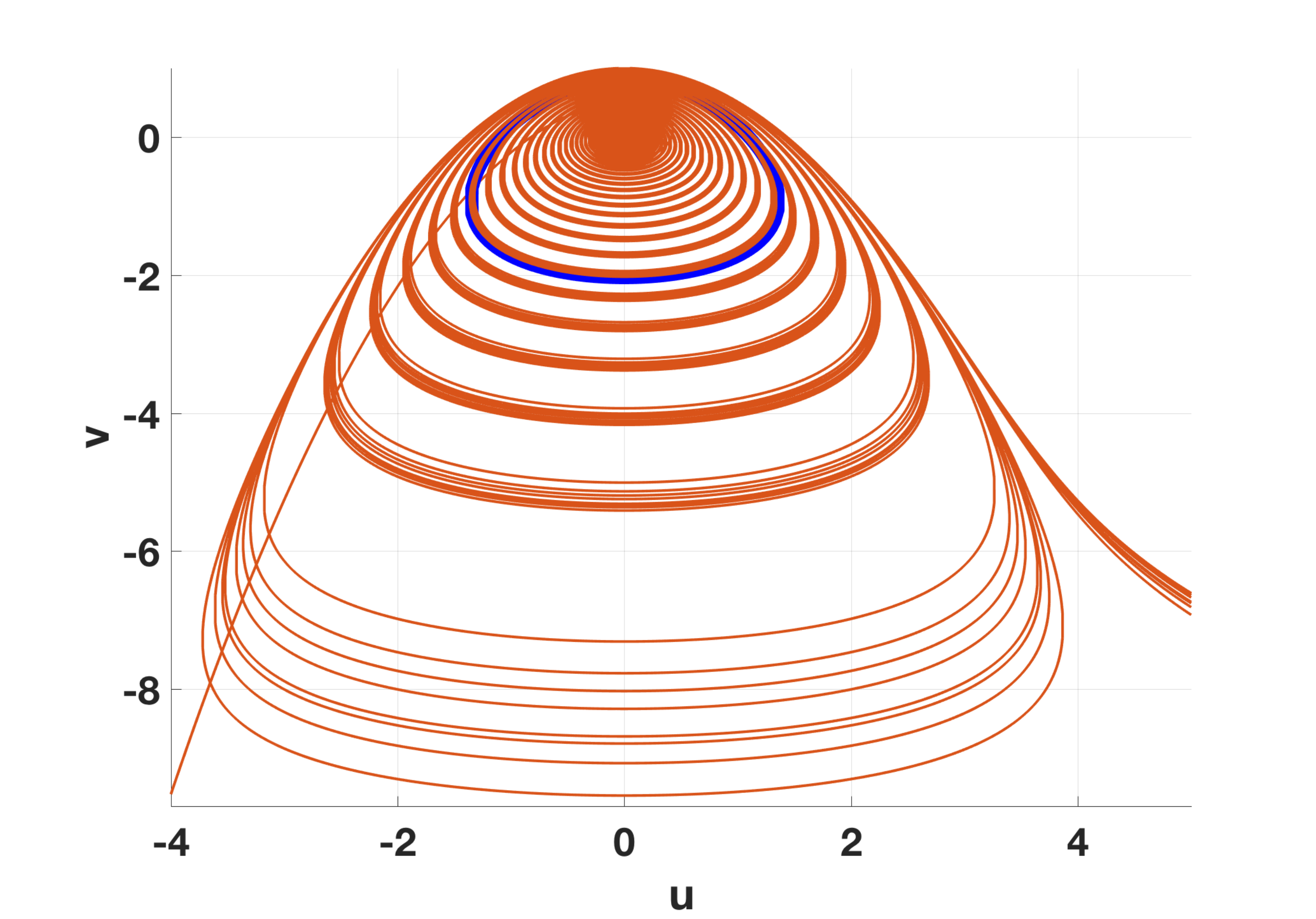}}
  \hspace{0.5cm}
      \subfloat[A trajectory lying in the ``basin of short transients".]{\includegraphics[width=7.5cm]{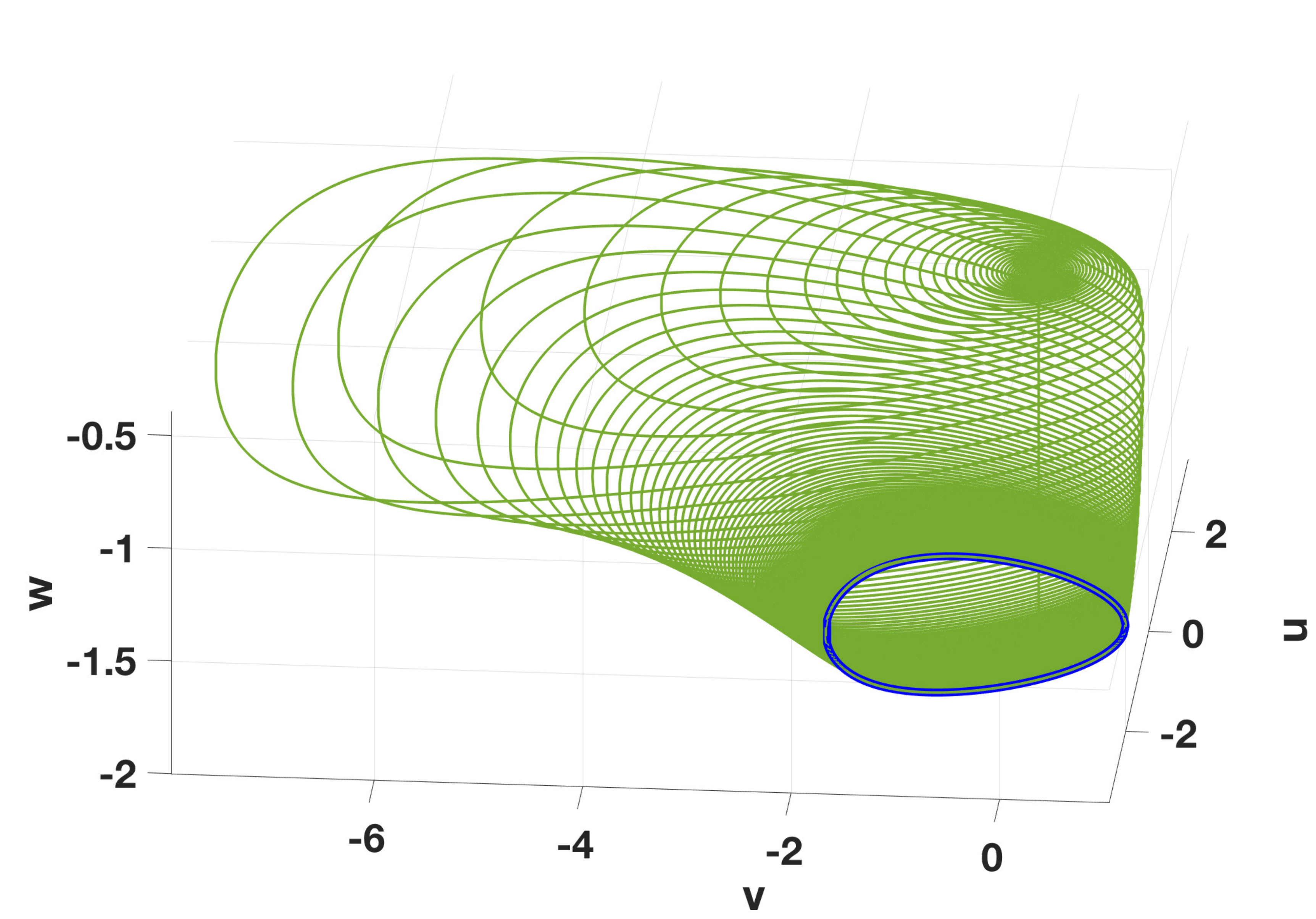}}
      \qquad
    \subfloat[Two trajectories: green lying in the ``basin of short transients" and cyan in the ``basin of long transients".]{\includegraphics[width=7.5cm]{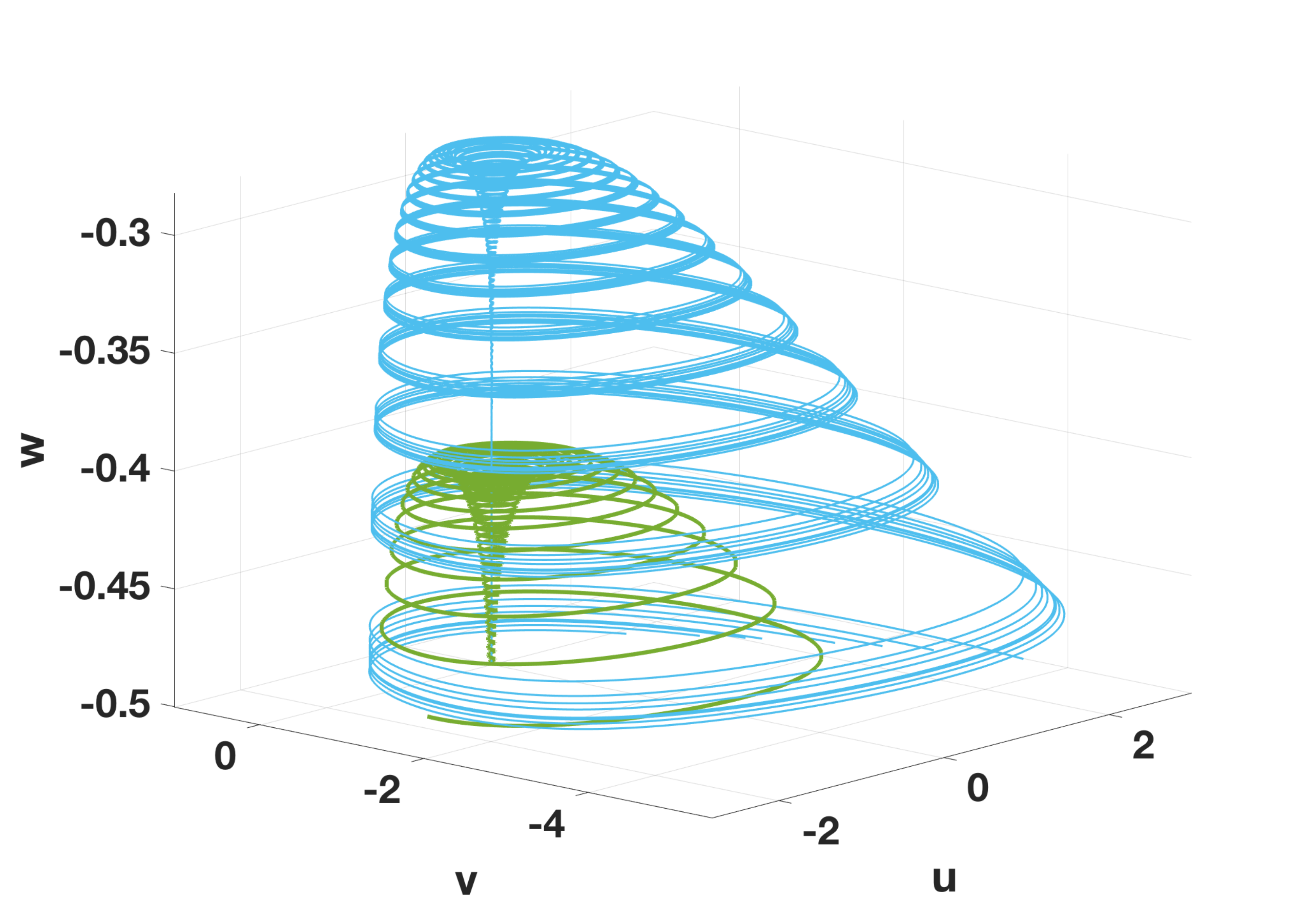}}
  \caption{(A): Phase portraits of the periodic orbit $\Gamma_{\alpha}$ (blue) and the transient MMO dynamics (orange) of system (\ref{normal2}) for $\alpha= 0.4613$ ($h=0.77898$), $\zeta =0.001$ and the other parameter values as in (\ref{parvalues01}). (B): A zoomed view  near the origin projected on the $uv$-plane. (C) A trajectory 
  (green)  approaching  $\Gamma_{\alpha}$ (blue) asymptotically.  (D): A zoomed view of the green trajectory in (C) as it directly approaches $\Gamma_{\alpha}$. The cyan trajectory performs a large excursion before it approaches $\Gamma_{\alpha}$.
  }
  \label{normal_form_beta1-25_bistability}
\end{figure}

The eigenvalues of the variational matrix of system (\ref{normal2}) at the equilibrium $p_e=(0, 0, 0)$ (which corresponds to the positive equilibrium of system (\ref{nondim3})) are $\lambda_1 \approx -0.0011$ and that  $\textnormal{Re}(\lambda_{2,3})>0$ if $0<\alpha<25.65$. A supercritical  Hopf bifurcation occurs at  $\alpha=\alpha_H=0$  (which corresponds to  $h_{HB}\approx 0.7787$ in system (\ref{nondim3})), where the first Lyapunov coefficient $l_1(0) \approx -1.0799<0$ can be calculated by Theorem \ref{hopfthm}. Consequently, a family of stable periodic orbits $\{\Gamma_{\alpha}\}$ is born at $\alpha_H$. System  (\ref{normal2})  approaches $\Gamma_{\alpha}$ asymptotically.

 We will consider $\alpha_c<\alpha<1$, where $\alpha_c \approx 0.34$, which is calculated by finding the root of $\tilde{\lambda}_{\delta}(\alpha)$ from (\ref{approxlamb}). For $\alpha$ in this range, the periodic attractor $\Gamma_{\alpha}$ is stable and complex dynamics such as MMOs occur as long transients before the system eventually asymptotes to $\Gamma_{\alpha}$ as shown in figure \ref{normal_form_beta1-25_bistability}(A)-(B).   We remark that though the normal form (\ref{normal2}) allows global returns of trajectories to the vicinity of the stable manifold of the equilibrium point, the global dynamics of (\ref{nondim3}) cannot be solely described by (\ref{normal2}). The presence of a return mechanism to a neighborhood of the equilibrium in (\ref{normal2}) makes it easier to visualize the transient dynamics.  To illustrate the dependence of the duration of transients on the initial conditions, two trajectories with different initial conditions are considered and their dynamics near the equilibrium $p_e$ are studied. The trajectories could directly approach $\Gamma_{\alpha}$ after they spiral out along $W^u(p_e)$ (see figure \ref{normal_form_beta1-25_bistability}(C)) or could perform a large excursion in phase space as seen in figure \ref{normal_form_beta1-25_bistability}(A) (compare with figure \ref {transient_amplitude} for the original system (\ref{nondim3})).

To detect an early warning sign of a large amplitude oscillation, two trajectories with similar local dynamics near $p_e$ are considered as shown  in figure \ref{normal_form_beta1-25_bistability}(D).  One of them lies in the basin of short transients, i.e. directly approaches $\Gamma_{\alpha}$, and the other  in the basin of long transients i.e. performs a large excursion before approaching $\Gamma_{\alpha}$, where the duration of transient is within $[0, 500]$. To detect which of the two trajectories directly approaches $\Gamma_{\alpha}$, we will apply the analysis from the previous section. 

  \begin{figure}[h!]     
  \centering 
{\includegraphics[width=9.5cm]{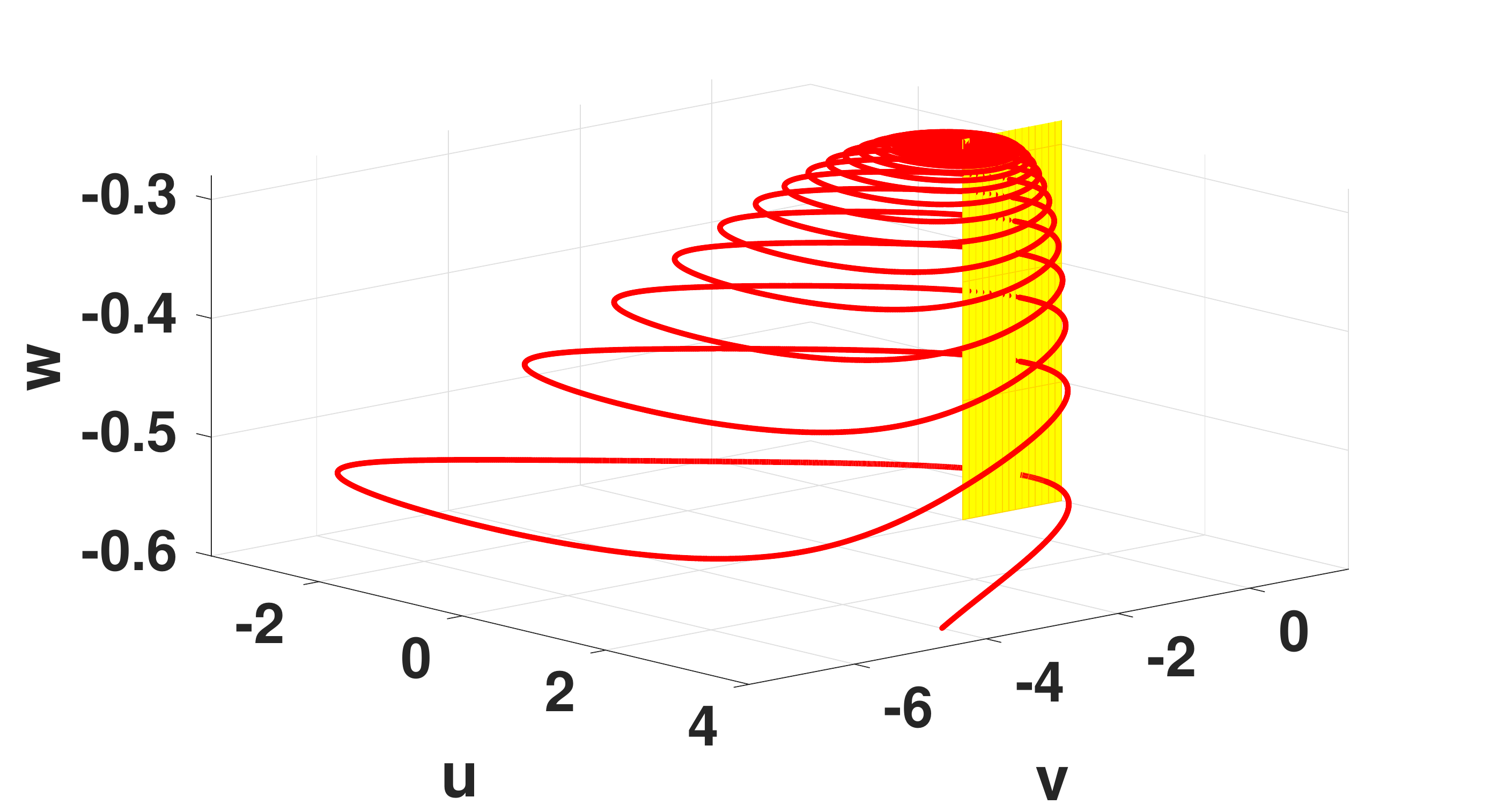}} 
  \caption{An illustration of  the Poincar\'e section $\{u=0\}$.}
  \label{section_illustration}
\end{figure}

As shown in figure \ref{normal_form_beta1-25_bistability}(D), during their  journeys towards $p_e$, the  trajectories  spiral up,  
reach their maximum heights (greater than $\lambda_c(\alpha)\approx -0.73$, where $\lambda_c(\alpha)$ is approximated by $\tilde{\lambda}_{\delta}(\alpha)$ (see Proposition \ref{approx_critthresh}) and then spiral out along $W^u(p_e)$.  As they descend, the SAOs grow in size and their intersections with the $\{u=0\}$ plane  in the increasing direction of $u$ are recorded. An illustration of the Poincar\'e section is shown in figure \ref{section_illustration}. The corresponding Poincar\'e return maps $v_1(\tau_1)$ (for the green trajectory)  and $v_2(\tau_1)$ (for the cyan trajectory) are shown in figure \ref{threshold_velocity}.  By Theorem \ref{thmdistint}, a trajectory exhibits a large amplitude oscillation if there exists some ${\tau}^+>0$ with $w({\tau}^+) >\lambda_c(\alpha)$ such that $v({\tau}^+)>v_{\alpha}^*(w({\tau}^+)) $ on the plane $\{u=0\}$. In other words, if the return map crosses the threshold curve $v_{\alpha}^*(\lambda)$ while $w(\tau)>\lambda_c(\alpha)$, then such a trajectory must perform a large excursion. The inset in  figure \ref{threshold_velocity} shows that the return map $v_2(\tau_1)$ goes above $v_{\alpha}^*(\lambda)$ for some $\tau_1$, whereas $v_1(\tau_1)$ stays below  $v_{\alpha}^*(\lambda)$ for all $\tau_1$. This in turn implies that the cyan trajectory in figure \ref{normal_form_beta1-25_bistability}(D) exhibits a large amplitude oscillation, whereas the green trajectory approaches the limit cycle.

  \begin{figure}[h!]     
  \centering 
{\includegraphics[width=11.5cm]{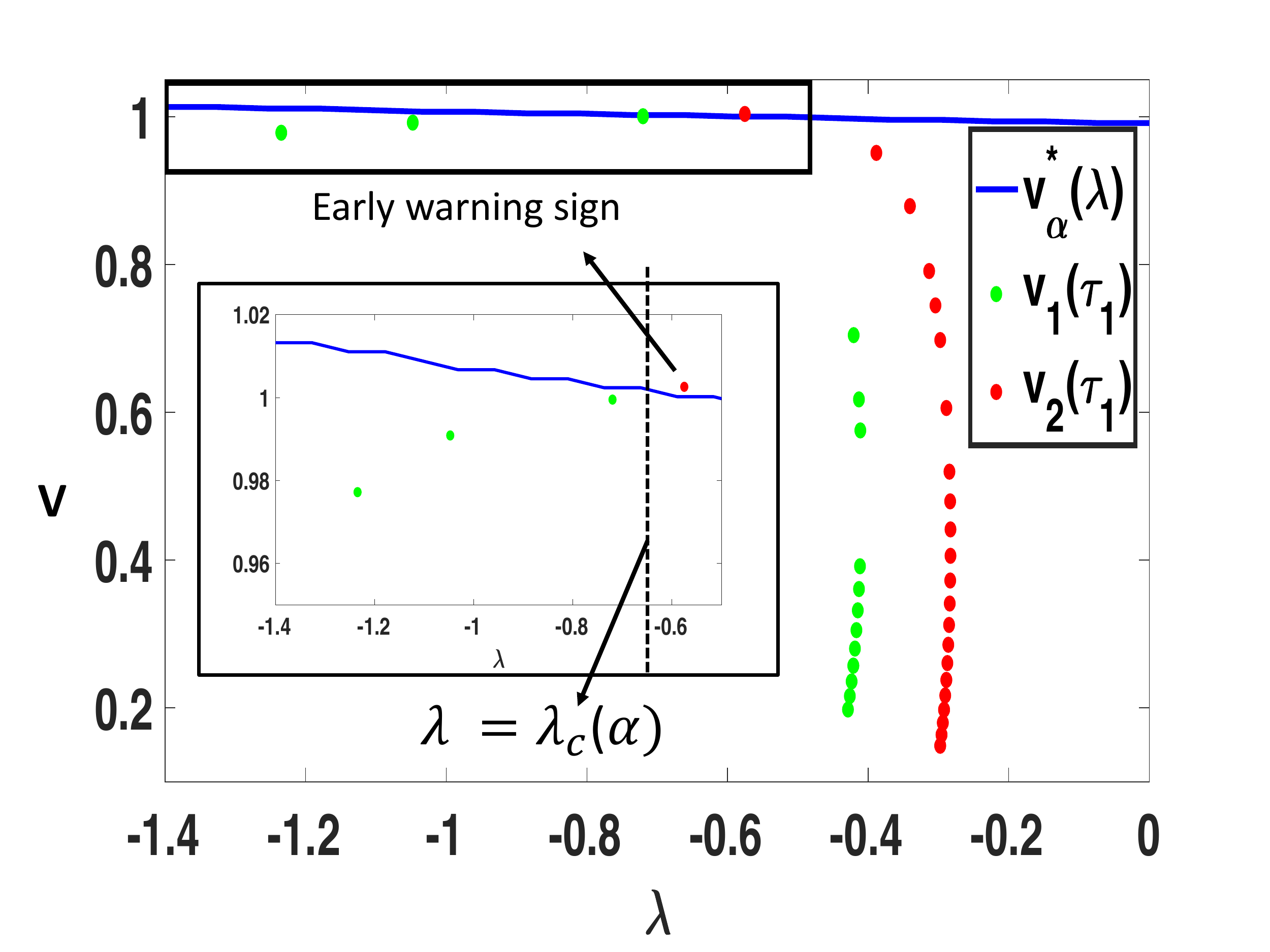}} 
  \caption{The threshold $v_{\alpha}^*(\lambda)$ and the Poincar\'e return maps $v_1(\tau_1)$ and $v_2(\tau_1)$  of the trajectories in figure \ref{normal_form_beta1-25_bistability}(D) on $\{u=0\}$. The inset gives a zoomed view of the Poincar\'e maps near $v_{\alpha}^*(\lambda)$. The return map $v_1(\tau_1)$ stays below $v_{\alpha}^*(\lambda)$ for all $\lambda$ whereas $v_2(\tau_1)$ goes above $v_{\alpha}^*(\lambda)$ at $\lambda\approx -0.58 $. Here $\lambda_c(\alpha) \approx -0.73$.   Note the monotonicity of the map $v_2(\tau_1)$ and the non-monotonicity of  $v_1(\tau_1)$.}
  \label{threshold_velocity}
\end{figure}


\section{Conclusion and outlook}

Sudden changes in ecological dynamics through time have been typically attributed to rapid response of dynamics to slow changes in environmental conditions such as climate change, habitat destruction and so forth \cite{SC}. However, sudden shifts may occur in a seemingly constant environment, and as proposed in \cite{Hastings1, Hastings2}, long transients can provide an alternate explanation for regime shifts in the absence of tipping points. The dynamics studied in this paper brings forth a linkage between empirical evidence to an ecological model which has the potential to mimic natural population cycles.

In this paper, we took a systematic approach to study long-term transient dynamics observed in a three-species predator-prey model with timescale separation. These transients occur as complex oscillatory patterns in form of MMOs near an FSN II singularity, also known as the singular-Hopf point. The transient MMOs eventually asymptote to a stable periodic orbit, born out of a nearby supercritical Hopf bifurcation, thereby demonstrating a  sudden population shift in a natural population. This paper focuses on a rigorous analysis of the underlying mechanism leading to such a transition and deriving a method of identifying an early warning sign of a sudden population change, which is highly desirable in ecological populations.

The sudden emergence of chaotic MMOs  past a supercritical bifurcation of the coexistence equilibrium state is intriguing.  A slow passage through a canard point along with influence of the local vector field around the saddle focus equilibrium are responsible for the long epochs of the SAOs in the transient MMO dynamics. The amplitudes of the SAOs are extremely small, which resemble the dynamics observed after a subcritical Hopf-homoclinic bifurcation in \cite{GWPW}.  However, we noted that  identifying early warning signs of a population outbreak from the phase space or the time series is very challenging. This is due to the fact that near the fold of the critical manifold $S$, whether a trajectory jumps to the other attracting branch of the critical manifold or gets trapped by the stable manifold of the Hopf limit cycle, it exhibits very similar local dynamics.  Hence, it seems that there does not exist a predictive criterion to infer an onset of  a large amplitude oscillation as a trajectory spends its time near the fold curve. 
 
 To accurately predict the onset of a large population fluctuation, the model was reduced to a suitable normal form near the FSN II bifurcation. We proved that the normal form possesses a separatrix, a boundary surface in the state space that separates two different types of oscillations. A large amplitude oscillation is initiated if a trajectory moves to the outer side of this surface, which is analogous to the idea of a quasi-separatrix crossing \cite{PKS}.  As a part of this analysis, it turns out that one needs to monitor the values of  the state variable $v$ as long as $w$ is greater than the threshold $\lambda_c(\alpha)$  to identify an early warning sign of a large fluctuation. The mechanism  of quasi-separatrix crossing  as the nonlinear dynamical mechanism responsible for spike initiation has been studied in many excitable slow-fast two-dimensional neural models (see \cite{DKR, PKS} and the references therein). In higher dimensions, explicitly finding  the quasi-separatrix can be difficult and has been recently studied in certain three-dimensional systems that can be locally reduced to planar dynamics with a drift \cite{ADKR}. It seems that the approach in this paper is novel in ecology and can be helpful for practical purposes.
 
 We also remark that a canonical form of the FSN II singularity was considered in \cite{KW}, where a blow-up analysis is performed to understand the evolution of MMOs in a small neighborhood of the FSN II singularity. The emergence of MMOs in \cite{KW} is attributed to the ``generalized canard phenomenon'', which is defined as a combination of a slow passage through a canard explosion, and a global return that resets the system dynamics after the passage has been completed \cite{BKW}.  The underlying principle generating the MMOs in this model is very closely related to the work in  \cite{KPK}.  
 

 An interesting area of investigation will be to study the effect of environmental fluctuations in  form of Weiner processes near the FSN II bifurcation.    In  slow-fast systems near an excitable regime  \cite{BL, S1}, a Markov-chain approach  can be taken to study the distribution of the random number of small oscillations between consecutive spikes. In  \cite{BL}, the authors proved that such a distribution  is asymptotically geometric with parameter related to the principal eigenvalue of a suitably defined substochastic Markov chain.  A similar approach was taken in a two-dimensional slow-fast stochastic predator-prey model to study the distribution of noise-induced oscillations near the singular Hopf bifurcation in \cite{S1}. A suitable stochastic normal form reduction near the FSN II singularity along with a Markov chain approach can be used to study the interspike time interval related to the random number of small amplitude oscillations separating consecutive large amplitude oscillations in this model. This in turn could then give insight to the distribution of return times of outbreaks.  In three or higher dimensions in a neighborhood of  singular Hopf bifurcation, these oscillations could be driven purely by noise or a combination of noise and bifurcations \cite{SK}, which makes this regime all the more interesting.  The interplay between noise-driven oscillations and deterministic SAOs of MMOs is another possible direction to explore.



\begin{remark}\label{rmk2} Most of the numerical simulations in this paper were done in MATLAB. We used the predefined routine ODE$45$  with relative and absolute error tolerances $10^{-11}$ and $10^{-12}$ respectively.

\end{remark}


\section{Appendix}
 \subsection{Center manifold reduction} The center manifold 
 can be expressed as a graph
\bess
w=\kappa(u, v, \delta)&=& \frac{1}{2}\kappa_{uu}^0 u^2+ \kappa_{uv}^0 uv+ \frac{1}{2}\kappa_{vv}^0 v^2 +O(3)+ \delta(\frac{1}{2}\kappa_{uu}^1 u^2+ \kappa_{uv}^1 uv+ \frac{1}{2}\kappa_{vv}^1 v^2 +O(3))\\
&+&O(\delta^2),
\eess 
where $O(3)$ represents cubic and higher-order terms in $u$ and $v$. The function $\kappa$ can be determined by solving the equation $\frac{d \kappa}{d \tau} = \delta(H_3 \kappa+\frac{1}{2}H_{11} u^2)+O(\delta^2)$. Using the equation for $w$ and the above equation and equating the coefficients of like terms,  (see \cite{BB} for details) one obtains that
\bess w=\kappa(u, v, \delta)=-\frac{H_{11}}{4H_{3}}(u^2+v^2) +O(3)+O(\delta).
\eess
The corresponding equations in the center manifold  up to higher order terms are
\begin{eqnarray}\label{centermfd}    
\left\{
\begin{array}{ll} \frac{du}{d\tau} &=v+\frac{u^2}{2}+\delta \Big(\alpha u +(-\frac{F_{13}H_{11}}{4H_{3}}+\frac{1}{6}F_{111})u^3-\frac{F_{13}H_{11}}{4H_3} uv^2 \Big)\\
   \frac{dv}{d\tau} &=-u.
       \end{array} 
\right. 
\end{eqnarray}

It is clear from the governing equations of the two-dimensional center manifold that the equilibrium $(0,0, 0)$ (up to $O(\delta)$) is asymptotically stable if and only if $\alpha<0$, where $\alpha$ is given by (\ref{del}). A Hopf bifurcation occurs at $\alpha=0$ and the first Lyapunov coefficient  \cite{GH, K1} is
\bess 
l_1(0)= \frac{\delta}{4}\Big(\frac{1}{2}F_{111}- \frac{F_{13}H_{11}}{H_3}\Big).
\eess

\subsection{Special Cases of the normal form} The FSN II point $(\bar{x}, \bar{y}, \bar{z}, \bar{h})$ can be explicitly computed for the following special cases:

{\em{Equal predation efficiencies}}: When $\beta_1=\beta_2$, one can solve for $\bar{h}$ in terms of the other parameters and can compute $(\bar{x}, \bar{y},\bar{z})$. In this case, 
\bess
\bar{h}=\frac{4\Big[(\alpha_{12}+\alpha_{21})(1-\beta_1)-(c\alpha_{21}+d\alpha_{12})(1+\beta_1)\Big] -\alpha_{12}\alpha_{21}(1+\beta_1)^3}{4(1-\beta_1-c(1+\beta_1))},
\eess
where $\alpha_{12}, \alpha_{21}, c,  d$ are free parameters that satisfy (\ref{cond1})-(\ref{cond2}) below:
\bes
\label{cond1} \alpha_{12}(1+\beta_1)^2  >4\Big(\frac{1-\beta_1}{1+\beta_1}-c\Big)& >& 0\\
 \label{cond2} \ \alpha_{21}(1+\beta_1)^2  -4(c-d)&\neq & 0.
\ees
Furthermore, 
\bess \bar{x} &= &\frac{1-\beta_1}{2},\\
 \bar{y} &=& \frac{\alpha_{12}(1+\beta_1)^3-4(1-\beta_1)+4c(1+\beta_1)}{4\alpha_{12}(1+\beta_1)},\\
  \bar{z} &=& \frac{1-\beta_1-c(1+\beta_1)}{\alpha_{12}(1+\beta_1)}.
\eess
For biological significance, in addition to  (\ref{cond1})- (\ref{cond2}) we will choose $\alpha_{12}, \alpha_{21}, c,  d$ so that that the expression in the numerator of $\bar{h}$ is positive, i.e.
\bess
4\Big[(\alpha_{12}+\alpha_{21})(1-\beta_1)-(c\alpha_{21}+d\alpha_{12})(1+\beta_1)\Big] > \alpha_{12}\alpha_{21}(1+\beta_1)^3.
\eess

{\em{No exclusive competition}}: Assuming that the predators do not exhibit interference competition, i.e. $\alpha_{12}=\alpha_{21}=0$, one can explicitly solve for $(\bar{x}, \bar{y},\bar{z}, \bar{h})$, namely,
\bess
\bar{x} &=& \frac{c\beta_1}{1-c},\\
\bar{y} &=& \frac{\beta_1^2}{(1-c)^3(\beta_1-\beta_2)}((1-c)(1-\beta_2)-2c\beta_1),\\
\bar{z} &=& \frac{(c\beta_1 +\beta_2(1-c))^2}{(1-c)^3(\beta_2-\beta_1)}((1-c)-\beta_1(1+c))
\eess
with 
\bes
\bar{h} =  \frac{(\beta_2-\beta_1)(\frac{\bar{x}}{\beta_2+\bar{x}} -d)}{(1-\beta_1-2\bar{x})(\beta_2+\bar{x})^2}.
\ees


\end{document}